\documentclass[numberwithinsect,a4paper,UKenglish,cleveref,autoref,thm-restate]{lipics-v2021}

\usepackage{
	amssymb,
	amsmath,
	amsthm,
	mathabx,
	stmaryrd,
	thm-restate,
} % Symbol and font packages

\usepackage{tikz} % Tikz
\usetikzlibrary{
	automata,
	positioning,
	arrows,
	shapes.geometric,
	calc
}

\tikzstyle{automaton} = [
	->,
	> = stealth,
	node distance = 2cm,
	every state/.style = {
		thick,
		fill=black!0.0,
		draw=black!0.0
	}
]

\usepackage{
	prftree,
	mathpartir
} % Alignment packages

\usepackage{hyperref}
\hypersetup{
    colorlinks=true,
    linkcolor=purple,
    filecolor=magenta,
    citecolor=magenta,
    urlcolor=cyan
}

\usepackage{xspace}

\newcommand{\sem}[1]{\left[\!\left[#1\right]\!\right]}
\newcommand{\trans}[2][]{\mathrel{\raisebox{-3pt}{$\xrightarrow[#1]{#2}$}}}

\newcommand{\filter}{\mathop{@}}

\newcommand \define[1]  {\emph{\textbf{#1}}} % definition statement plus indexing

\newcommand	\N 			{{\mathbb N}} % natural numbers
\newcommand \Exp        {{\operatorname{Exp}}} % expressions
\newcommand \acro[1] 	{\(\mathsf{#1}\)}
\newcommand \GKAT 		{\acro{GKAT}\xspace}

\newcommand \Cov		{{\operatorname{Cov}}} % for covarieties
\newcommand \coalg[1]	{{\CMcal #1}} % for coalgebras
\newcommand \coeq[1]	{{\mathsf #1}} % for coequations

\newcommand \bisim 		{\mathbin{\underline{\leftrightarrow}}} % bisimilar
\newcommand \V 			{{\coalg V}} % generic coalgebra
\newcommand \X 			{{\coalg X}} % generic coalgebra
\newcommand \Y 			{{\coalg Y}} % generic coalgebra
\newcommand \Z 			{{\coalg Z}} % generic coalgebra
\newcommand  \R  		{{\coalg R}} % bisimulation relation

\newcommand \Node 		{{\operatorname{Node}}} % nodes of an ordered tree
\newcommand \dom 		{{\operatorname{dom}}} % function domain

\newcommand \nin        {\mathbin{\not\in}} % not in

\newcommand \img 		{{\operatorname{img}}}
\newcommand \gequiv 	{\mathbin{\dot \equiv}}
\makeatletter
\newcommand{\customlabel}[2]{%
   \protected@write \@auxout {}{\string \newlabel {#1}{{#2}{\thepage}{#2}{#1}{}} }% % chktex 1
   \hypertarget{#1}{#2}%
}
\makeatother

\title{
  Guarded Kleene Algebra with Tests: \texorpdfstring{\\}{}
  Coequations, Coinduction, and Completeness
}

\titlerunning{GKAT:\@
  Coequations, Coinduction, and Completeness} %TODO optional, please use if title is longer than one line

\author{Todd Schmid}{Department of Computer Science, University College London, United Kingdom}{todd.schmid.19@ucl.ac.uk}{https://orcid.org/0000-0002-9838-2363}{}
\author{Tobias Kapp\'e}{Department of Computer Science, Cornell University, Ithaca, NY, USA}{tkappe@cornell.edu}{https://orcid.org/0000-0002-6068-880X}{DARPA grant HR001120C0107 (Pronto)}
\author{Dexter Kozen}{Department of Computer Science, Cornell University, Ithaca, NY, USA}{kozen@cs.cornell.edu}{https://orcid.org/0000-0002-8007-4725}{NSF grant CCF-2008083}
\author{Alexandra Silva}{Department of Computer Science, University College London, United Kingdom}{alexandra.silva@ucl.ac.uk}{https://orcid.org/0000-0001-5014-9784}{ERC Consolidator Grant AutoProbe (101002697) and a Royal Society Wolfson Fellowship}

\authorrunning{T. Schmid, T. Kapp\'e, D. Kozen, A. Silva} %TODO mandatory. First: Use abbreviated first/middle names. Second (only in severe cases): Use first author plus 'et al.'

\Copyright{Todd Schmid, Tobias Kapp\'e, Dexter Kozen, and Alexandra Silva} %TODO mandatory, please use full first names. LIPIcs license is "CC-BY";  http://creativecommons.org/licenses/by/3.0/

 %TODO mandatory; please add comma-separated list of keywords

\category{Track B: Automata, Logic, Semantics, and Theory of Programming}

\relatedversion{}\relatedversiondetails{Full Version}{https://arxiv.org/abs/2102.08286} %optional, e.g. full version hosted on arXiv, HAL, or other respository/website
\nolinenumbers %uncomment to disable line numbering

\hideLIPIcs  %uncomment to remove references to LIPIcs series (logo, DOI, ...), e.g. when preparing a pre-final version to be uploaded to arXiv or another public repository

\EventEditors{Nikhil Bansal, Emanuela Merelli, and James Worrell}
\EventNoEds{3}
\EventLongTitle{48th International Colloquium on Automata, Languages, and Programming (ICALP 2021)}
\EventShortTitle{ICALP 2021}
\EventAcronym{ICALP}
\EventYear{2021}
\EventDate{July 12--16, 2021}
\EventLocation{Glasgow, Scotland (Virtual Conference)}
\EventLogo{}
\SeriesVolume{198}
\ArticleNo{141}
\begin{document}

\maketitle

\begin{abstract}
    Guarded Kleene Algebra with Tests (\acro{GKAT}) is an efficient fragment of \acro{KAT}, as it allows for almost linear decidability of equivalence.
    In this paper, we study the (co)algebraic properties of \acro{GKAT}. %  and show how these can be used in completeness proofs. % chktex 36
    Our initial focus is on the fragment that can distinguish between unsuccessful programs performing different actions, by omitting the so-called \emph{early termination axiom}.
    We develop an operational (coalgebraic) and denotational (algebraic) semantics and show that they coincide.
    We then characterize the behaviors of \acro{GKAT} expressions in this semantics, leading to a coequation that captures the covariety of automata corresponding to these behaviors.
    Finally, we prove that the axioms of the reduced fragment are sound and complete w.r.t.\ the semantics, and then build on this result to recover a semantics that is sound and complete w.r.t.\ the full set of axioms.
\end{abstract}

\section{Introduction}
\emph{Kleene algebra with tests} (\acro{KAT})~\cite{kat} was introduced in the early 90's as an extension of Kleene algebra (\acro{KA}), the algebra of regular expressions.
The core idea of the extension was simple: consider regular languages over a two-sorted alphabet, in which one sort represents Boolean tests and the other denotes basic program actions.
This seemingly simple extension enables an important application for regular languages in reasoning about imperative programs with basic control flow structures like branches (\textbf{if}-\textbf{then}-\textbf{else}) and loops (\textbf{while}).
\acro{KAT} largely inherited the properties of \acro{KA}:\@ a language model~\cite{kozen-smith-1996}, a Kleene theorem~\cite{kozen-2003}, a sound and complete axiomatization~\cite{kozen-smith-1996}, and a \textsc{pspace} decision procedure for equivalence~\cite{cohen-kozen-smith-1996}.

In 2014, a specialized \acro{KAT} called \acro{NetKAT}~\cite{netkat} was proposed to program software-defined networks.
%Equivalence checking of \acro{NetKAT} programs serves as a basis to verify important properties such as reachability.
\acro{NetKAT} was later extended with a probabilistic choice operator that enabled the modelling of randomized protocols~\cite{probnetkat}.
Interestingly, there exists a decision procedure for \acro{NetKAT} program equivalence that enables practical verification of reachability in networks with thousands of nodes and links, which seems to scale almost linearly despite the \textsc{pspace}-completeness of this problem~\cite{netkat-algo,mcnetkat}.
This raised the question: do practical \acro{NetKAT} programs belong to a fragment of \acro{KAT} that has more favorable properties than the full language?

Recently, this question was answered positively~\cite{gkat}, in the form of \emph{Guarded Kleene Algebra with Tests} (\acro{GKAT}), a fragment of \acro{KAT} obtained by adding a Boolean guard to the non-deterministic choice and iteration operators so that they correspond exactly to the standard \textbf{if-then-else} and \textbf{while} constructs.
\acro{GKAT} is expressive enough to capture all programs used in network verification while allowing for almost linear time\footnote{$O(n\alpha(n))$, where $\alpha(n)$ is the inverse of Ackermann's function} decidability of equivalence, thereby explaining the experimental results observed in \acro{NetKAT}. % chktex 8

The use of \acro{GKAT} as a framework for program analysis also raises further questions about recovering the properties of \acro{KAT} on the level of \acro{GKAT}.
Is there a class of automata that provides a Kleene theorem?
Is there a sound and complete axiomatization of \acro{GKAT} equivalence?
The original paper~\cite{gkat} gave incomplete answers to these questions.
First, it proposed a class of \emph{well-nested} automata that can be used to describe the semantics of all \acro{GKAT} programs, but left open whether this class covered all automata that accept the behaviors of \acro{GKAT} programs.
Second, \acro{GKAT} was axiomatized under the assumption of \emph{early termination}: intuitively, referring to a semantics of imperative programs where programs that fail immediately are equated to programs that fail eventually.
This semantics, though useful, is too coarse in contexts where program behavior prior to failure matters.

In this paper, we take a new perspective on the semantics of \acro{GKAT} programs and their corresponding automata, using coequations.
Coequations provide the right tool to characterize fragments of languages as they enable a precise way to remove unwanted traces.
We are then able to give a precise characterization of the behaviors of \acro{GKAT} programs and prove a completeness theorem for each of the fragments of interest.

\medskip
\noindent\textbf{Our contributions.}
In a nutshell, the contributions of this paper are the following:
\begin{enumerate}%[leftmargin=*]
\item
	We give a denotational model for \GKAT without early termination by representing the behavior as a certain kind of tree.
	This allows us to design two coequations: one characterizing the behaviors denoted by \acro{GKAT} expressions, and another capturing only the behaviors of \acro{GKAT} expressions that terminate early.
\item
	We obtain two completeness results for \acro{GKAT}: one for the model of the previous item and the axiomatization of~\cite{gkat} without the early termination axiom; and building on this, another for the full axiomatization.
	The former is new; the latter provides an alternative proof to the completeness theorem presented in~\cite{gkat}.
\item
	A concrete example of a well-nested \acro{GKAT} automaton with a non-well-nested quotient.
	This settles an open question of~\cite{gkat} and closes the door on an alternative proof of completeness based on well-nested automata.
\end{enumerate}

% \textbf{Acknowledgements.} The authors would like to thank the reviewers for their helpful comments. We would also like to thank Jurriaan Rot for his insightful contributions to early discussions on this project.

% *** *** SECTION *** *** %
\section{Guarded Kleene Algebra with Tests}\label{sec:guarded_kleene_algebra_with_tests}

At its heart, \emph{Guarded Kleene Algebra with Tests} (\acro{GKAT}) is an algebraic theory of imperative programs.
Expressions in \acro{GKAT} are concise formulas for \textsc{while} programs~\cite{kozentseng2008}, which are built inductively from actions and tests with sequential composition and the classic programming constructs of branches and loops: $\textbf{if } b \textbf{ then } e \textbf{ else } f$ and $\textbf{while } b \textbf{ do } e$.

Formally, these expressions are drawn from a two-sorted language of \emph{tests} and \emph{programs}.
The tests are built from a finite set of \define{primitive tests} \(T\), as follows:
\begin{mathpar}
	\operatorname{BExp} \ni b,c ::= 0 \mid 1 \mid t \in T \mid \bar b \mid b \wedge c \mid b \vee c.
\end{mathpar}
Here, $0$ and $1$ are understood as the constant tests \textbf{false} and \textbf{true} respectively, $\bar{b}$ denotes the negation of $b$, and $\wedge$ and $\vee$ are conjunction and disjunction, respectively. We will use \(A\) to denote the set of \define{atomic tests} (or just \define{atoms}), Boolean expressions of the form
\(
	d_1 \wedge \dots \wedge d_l,
\)
 where \(d_i \in \{t_i, \bar t_i\}\) for each \(i \le l\) and \(\{t_i \mid i \le l\}\) is a fixed enumeration of \(T\). It is well known that any \(b \in \operatorname{BExp}\) can be written equivalently as the disjunction of the atoms \(a \in A\) that imply \(b\) under the laws of Boolean algebra. We will often identify each Boolean expression \(b \in \operatorname{BExp}\) with this set of atoms and write \(b \subseteq A\) or \(a \in b\).

Programs are built from tests and a finite set of \define{primitive programs} or \define{actions} \(\Sigma\), disjoint from \(T\).
Formally, programs are generated by the grammar
\begin{mathpar}
	 \Exp \ni e, f ::= b \in \operatorname{BExp} \mid p \in \Sigma \mid e \cdot f \mid e +_b f \mid e^{(b)}
\end{mathpar}
Here, a test $b$ abbreviates the statement \textbf{assert $b$}, the operator $\cdot$ is sequential composition, $e +_b f$ is shorthand for \textbf{if $b$ then $e$ else $f$} and $e^{(b)}$ is shorthand for \textbf{while $b$ do $e$}.

\acro{GKAT} programs satisfy standard properties of imperative programs.
For instance, swapping the branches of an \textbf{if-then-else} construct should not make a difference, provided that we also negate the condition; that is, the semantics of \(e +_b f\) should coincide with that of \(f +_{\overline{b}} e\). %, regardless of the programs \(e\) and \(f\) or the test \(b\).
The rules in \cref{fig:GKAT axioms} axiomatize equivalences between programs.
Together with the axioms of Boolean algebra, these generate a congruence \(\equiv\) on \(\Exp\).

\begin{figure}
\small
\begin{tabular}{l@{~} >{$}r<{$}@{~} >{$}c<{$}@{~} >{$}l<{$}
                l@{~} >{$}r<{$}@{~} >{$}c<{$}@{~} >{$}l<{$}
                l@{~} >{$}r<{$}@{~} >{$}c<{$}@{~} >{$}l<{$}}
\multicolumn{4}{l}{\hspace*{-1ex}\textbf{Union Axioms}}
    & \multicolumn{4}{l}{\hspace*{-1ex}\textbf{Sequence Axioms}}
    & \multicolumn{4}{l}{\hspace*{-1ex}\textbf{Loop Axioms}} \\
\customlabel{ax:idemp}{U1}. &e+_b e            &\equiv &e   & % chktex 1
   \customlabel{ax:seqassoc}{S1}. & (e \cdot f) \cdot g &\equiv  &e \cdot (f \cdot g) & % chktex 1
   \customlabel{ax:unroll}{W1}. &e^{(b)} &\equiv  & e \cdot e^{(b)} +_{b} 1 \\ % chktex 1
\customlabel{ax:skewcomm}{U2}. &e +_b f           &\equiv &f +_{\bar{b}} e & % chktex 1
   \customlabel{ax:absleft}{S2}. &0 \cdot e           &\equiv  &0 & % chktex 1
   \customlabel{ax:tighten}{W2}. & {(ce)}^{(b)} &\equiv& {(e +_c 1)}^{(b)} \\
\customlabel{ax:skewassoc}{U3}. & (e +_b f) +_c g   &\equiv &e +_{b \wedge c} (f +_c g) & % chktex 1
   \customlabel{ax:absright}{S3}. &e \cdot 0           &\equiv  &0 & % chktex 1
   \multirow{3}{*}{\customlabel{ax:fixpoint}{W3}.} & \multicolumn{3}{c}{\multirow{3}{*}{\(\inferrule{E(e) \equiv 0 \\ \hspace{-3ex}g \equiv eg +_b f}{g \equiv e^{(b)} \cdot f}\)}}  \\
\customlabel{ax:guard-if}{U4}. &e +_b f           &\equiv &b \cdot e +_b f & % chktex 1
   \customlabel{ax:neutrleft}{S4}. & \multicolumn{3}{c}{\(1 \cdot e \equiv e\); S5. \(e \equiv e \cdot 1\)} \\ % chktex 1
\customlabel{ax:rightdistr}{U5}. &e \cdot g +_b f \cdot g &\equiv & (e +_b f) \cdot g & % chktex 1
   \customlabel{ax:neutrright}{S6}. &b \cdot c           &\equiv  & b \wedge c \\[0.5em] % chktex 1
\end{tabular}%
\caption{Axioms for \acro{GKAT}-expressions. Here, \(e,f,g \in \Exp\) and \(b,c \in \operatorname{BExp}\).}\label{fig:GKAT axioms}
\end{figure}

Some remarks are in order for axiom W3. %, which is an axiom schema, rather than a singular axiom.
The right-hand premise states that an expression \(g\) has some self-similarity in the sense that it is equivalent to checking whether \(b\) holds, in which case it runs \(e\) followed by recursing at \(g\), and otherwise running \(f\).
Intuitively, this says that \(g\) is loop-like, matching the conclusion that \(g\) is equivalent to \(e^{(b)} \cdot f\).
However, this conclusion may not make sense when based on just the second premise.
Specifically, if we choose \(e\), \(f\), \(g\) and \(b\) to be \(1\), we can show that the premise holds and derive \(1 \equiv 1^{(1)} \cdot 1\), which is to say that \textbf{assert true} is equivalent to \textbf{(while true do assert true); assert true}.
Intuitively, this should be false: the first program terminates successfully and immediately, but the second program does not.
The problem is that the loop body does not perform any actions that affect the state and make progress towards the end of the loop.

This is remedied by the left-hand premise, which distinguishes loop bodies that can accept immediately from those that cannot.
It plays the same role as the \emph{empty word property} in Salomaa's axiomatization of the algebra of regular events~\cite{salomaa1966}. %It is also related to the notion of \emph{productivity} in coalgebraic proof systems.
Formally, given \(e \in \Exp\), the Boolean expression \(E(e)\) is defined inductively by setting \(E(p) = 0\), \(E(b) = b\), and
\begin{align*}
	E(e\cdot f) = E(e) \wedge E(f)\qquad
	E(e +_b f) = (b \wedge E(e)) \vee (\bar b \wedge E(f))
	 \qquad E(e^{(b)}) = \bar b
\end{align*}
We call \(e\) \define{productive} if \(E(e) \equiv 0\).
Axioms W2 and W3 are analogues of Salomaa's axioms A\(_{11}\) and R2~\cite{salomaa1966}.
Specifically, W2 says that non-productive loop iterations do not contribute to the semantics.
This allows the use of W3 to reason about loops in general, for instance to prove \( e^{(b)} \equiv e^{(b)} \cdot \overline{b} \), which says that the loop condition is false when a loop ends~\cite{gkat}.

\smallskip
Axiom S3 identifies a program that fails eventually with the program that fails immediately.
As a consequence, \(\equiv\) cannot distinguish between processes that loop forever, like \(p^{(1)}\) and \(q^{(1)}\), even though they perform different actions~\cite{gkat}.
Consequently, \acro{GKAT} can be seen as a theory of \emph{computation} schemata, i.e., programs that need to halt successfully to be meaningful.

In contrast, it is also useful to be able to reason about \emph{process} schemata, i.e., programs that perform meaningful tasks, even when they do not terminate successfully.
To this end, we define the \define{reduced congruence \(\equiv_0\)} generated by the axioms of \cref{fig:GKAT axioms} except S3.

\medskip
Let \(\sem{-}: \Exp \to S\) be a semantics of \acro{GKAT}\@. We say that $\sem-$ is \define{sound w.r.t.~\(\equiv\)} if for all \(e, f \in \Exp\) with \(e \equiv f\), it holds that \(\sem{e} = \sem{f}\).
Similarly, \(\sem{-}\) is \define{sound w.r.t.~\(\equiv_0\)} if \(e \equiv_0 f\) implies that \(\sem{e} = \sem{f}\).

Since \(\equiv\) encodes common program laws, one might wonder whether there is a single interpretation in which programs are related by \(\equiv\) if and only if they have the same image.
Such an interpretation is called \define{free w.r.t.~\(\equiv\)}.
This question is not just of theoretical interest: a free interpretation can help decide whether programs are provably equivalent, and hence the same under any sound interpretation, by checking whether their free semantics coincide.
Naturally, the same question can be asked for \(\equiv_0\): is there a semantics that is \define{free w.r.t.~\(\equiv_0\)}, i.e., where \(e \equiv_0 f\) if and only if \(e\) and \(f\) have the same interpretation?

\medskip
The remainder of this paper is organized as follows.
In \cref{sec:coalgebra}, we describe the operational structure for \acro{GKAT} expressions in terms of \GKAT-automata, as in~\cite{gkat}.
In \cref{sec:the_final_gkat-automaton}, we provide an explicit construction of a \acro{GKAT}-automaton in which all other automata can be uniquely interpreted.
We then build a semantics that is sound w.r.t.~\(\equiv_0\) in \cref{sec:i_g_as_an_algebra}.
In \cref{sec:well_nested_coalgebras} we relate our coequational description of \acro{GKAT} expressions to the \emph{well-nested \GKAT-automata} of~\cite{gkat}.
In \cref{sec:completeness}, we prove that this semantics is in fact complete w.r.t.~\(\equiv_0\) and, building on this, obtain a semantics that is complete w.r.t.~\(\equiv\).
Omitted proofs are included in the appendix.

% *** subsection *** %
\section{An operational model: \texorpdfstring{\acro{GKAT}}{GKAT}-automata}%
\label{sec:coalgebra}

In this section we discuss the small-step operational model for \acro{GKAT} programs from~\cite{gkat}.
The operational perspective provides us with the tools to describe a semantics that is complete w.r.t.~\(\equiv_0\) and paves the way to a decision procedure.

We can think of a \acro{GKAT}-program as a machine that evolves as it reads a string of atomic tests.
Depending on the most recently observed atomic test, the program either accepts, rejects, or emits an action label and changes to a new state.
For example, feeding \textbf{if \(b\) do \(p\) else \(q\)} an atomic test \(a \in b\) causes it to perform the action \(p\) and then terminate successfully.

\begin{definition}
A \define{\acro{GKAT}-automaton}~\cite{gkat,kozentseng2008} is a pair \(\X = (X,\delta)\), where \(X\) is a set of \define{states} and \(\delta: X \times A \to 2 + \Sigma \times X\) is a \define{transition function}.
We use \(x \trans{a|p}_\X x'\) as a notation for \(\delta(x, a) = (p,x')\).
Similarly, \(x \Rightarrow_\X a\) denotes that \(\delta(x, a) = 1\), and \(x \downarrow_\X a\) denotes that \(\delta(x, a) = 0\).
We drop the subscript \(\X\) when the automaton is clear from context.
\end{definition}
Intuitively, \(X\) represents the states of an abstract machine running a \acro{GKAT} program, with dynamics encoded in \(\delta\).
When the machine is in state \(x \in X\) and observes \(a \in A\), there are three possibilities: if \(x \downarrow a\), the machine rejects; if \(x \Rightarrow a\), it accepts; and if \(x \trans{a|p} x'\), it performs the action \(p\) followed by a transition to the state \(x'\).

\begin{remark}\label{rem:all the coalgebra in one remark}
	The reader familiar with coalgebra will recognize that \acro{GKAT}-automata are precisely coalgebras for the functor \(G = {(2 + \Sigma \times \mathsf{Id})}^A\)~\cite{gkat}.
	Indeed, the notions relating to \acro{GKAT}-automata, such as homomorphism, bisimulation, and semantics to follow are precisely those that arise from \(G\) as prescribed by universal coalgebra~\cite{Rutten2000}.
\end{remark}

We can impose an automaton structure on \(\Exp\) yielding the \define{syntactic \acro{GKAT}-automaton \(\coalg{E} = (\Exp, D)\)}, where \(D\) is the transition map given by Brzozowski derivatives~\cite{gkat} as specified in \cref{tab:transition_structure_of_coalg_exp}. For instance, the operational behavior of \(p^{(b)}\) as a state of \(\coalg E\) could be drawn as follows, where \(x \trans{b\mid p} y\) denotes that \(x \trans{a\mid p} y\) for every \(a \in b\) and rejecting transitions \(x \downarrow a\) are left implicit:
\begin{equation}%
\vspace*{-.2cm}
\label{example:p-loop}
\begin{tikzpicture}[
	baseline=3mm,
	->,
	> = stealth,
	every node/.style = {
		thick,
		minimum height=1em,
		text height=0.4em,
	},
]
		\node (0) {\footnotesize\(\bar{b}\)};
		\node[right=5mm of 0] (1) {\(p^{(b)}\)};
		\node[right=7mm of 1] (2) {\(1 \cdot p^{(b)}\)};
		\node[right=5mm of 2] (3) {\footnotesize\(\bar{b}\)};
		\path (1) edge[above,->] node {\footnotesize\(b | p\)} (2);
		\path (2) edge[loop above,->] node[right=.75mm,yshift=-1.3mm] {\footnotesize \(b | p\)} (2);
		\path (2.east) edge[double,double distance=2pt,-implies] (3.west);
		\path (1.west) edge[double,double distance=2pt,-implies] (0.east);
\end{tikzpicture}
\vspace*{-.2cm}
\end{equation}
% In \cref{example:p-loop}, the notation \(p^{(b)} \trans{b\mid p} 1 \cdot p^{(b)}\) is used to indicate that \(p^{(b)} \trans{a\mid p} 1 \cdot p^{(b)}\) for all \(a \in b\).

\begin{figure*}[!t]\small
	\begin{mathpar}
		\infer{a \in b}
		{b \Rightarrow a}
		\quad
		\infer{\ }
		{p \trans{a|p} 1}
		\quad
		\infer{a \in b \and e \Rightarrow a}
		{e +_b f \Rightarrow a}
		\quad
		\infer{a \in \bar b \and f \Rightarrow a}
		{e +_b f \Rightarrow a}
		\quad
		\infer{a \in b \and e \trans{a|p} e'}
		{e +_b f \trans{a|p} e'}
		\quad
		\infer{a \in \bar b \and f \trans{a|p} f'}
		{e +_b f \trans{a|p} f'}\\
		\infer{e \Rightarrow a \and f \Rightarrow a}
		{e \cdot f \Rightarrow a}
		\quad
		\infer{e \Rightarrow a \and f \trans{a|p} f'}
		{e \cdot f \trans{a|p} f'}
		\quad
		\infer{e \trans{a|p} e'}
		{e \cdot f \trans{a|p} e' \cdot f}
		\quad
		\infer{a \in b \and e \trans{a|p} e'}
		{e^{(b)} \trans{a|p} e' \cdot e^{(b)}}
		\quad
		\infer{a \in \bar b}
		{e^{(b)} \Rightarrow a}
	\end{mathpar}
	\vspace{-1em}
	\caption{The transition structure of \(\coalg{E}\).
	Here, \(e,e',f,f' \in \Exp\), \(b \subseteq A\), \(a \in A\), and \(p \in \Sigma\).
	Transitions that are not explicitly defined above are assumed to be failed termination.}%
	\label{tab:transition_structure_of_coalg_exp}\vspace*{-0.25cm}
\end{figure*}

\noindent The operational structure of \(\coalg E\) is connected to \(\equiv_0\) as follows.
\begin{theorem}[Fundamental theorem of GKAT]
For any \(e \in \Exp\),
\(
	e \equiv_0 1 +_{E(e)} D(e)
\)
where

{\vspace{-0.8em}\small
\begin{mathpar}
	D(e) =\!\!\! \bigplus_{\ \ {e \trans{a|p_a} e_a}}\!\!\! p_a \cdot e_a
	\qquad \text{ and } \qquad
	\bigplus_{a \in b} e_{a} = \begin{cases}
	0 &\text{if \(b = 0\)},\\
	e_{a} +_a \left(\bigplus\limits_{a' \in b\setminus a} e_{a'}\right) &\text{some $a\in b$, otherwise}.
\end{cases}
\end{mathpar}}
\end{theorem}
The generalized guarded union above is well defined, in that the order of atoms does not matter up to \(\equiv_0\).
See~\cite{gkat} for more details about the generalised guarded union.

\smallskip
States of \GKAT-automata have the same behavior if reading the same sequence of atoms leads to the same sequence of actions, acceptance, or rejection.
This happens when one state mimics the moves of the other, performing the same actions in response to the same stimuli.
For instance, consider the \acro{GKAT}-automaton in~\eqref{example:p-loop}: the behavior of \(p^{(b)}\) can be replicated by the behavior of \(1 \cdot p^{(b)}\), in that both either consume an \(a \in \bar{b}\) and terminate or consume \(a \in b\) and emit \(p\) before transitioning to \(1 \cdot p^{(b)}\).
This can be made precise.

\begin{definition}\label{lem:concrete bisimilarity}
	Let \(R \subseteq X \times Y\) be a relation between the state spaces of \acro{GKAT}-automata \(\X\) and \(\Y\).
	Then \(R\) is a \define{bisimulation} if for any \((x,y) \in R\) and \(a \in A\),
    {\setlength{\leftmargini}{13mm}
	\begin{itemize}
		\item[(1)] \(x \downarrow_\X a\) if and only if \(y \downarrow_\Y a\); and (2) \(x \Rightarrow_\X a\) if and only if \(y \Rightarrow_\Y a\); and
		\item[(3)] if \(x \trans{a|p}_\X x'\) and \(y \trans{a|q}_\Y y'\) for some \(x'\) and \(y'\), then \(p = q\) and \((x', y') \in R\).
	\end{itemize}}
    \noindent
	If a pair of states \((x,y) \in X \times Y\) is contained in a bisimulation, we say that \(x\) and \(y\) are \define{bisimilar}.
	If a bisimulation \(R\) is the graph of a function \(\varphi : X \to Y\), we write \(\varphi : \X \to \Y\) and call \(\varphi\) a \define{\acro{GKAT}-automaton homomorphism}~\cite{Rutten2000}.
\end{definition}

Indeed, bisimulations are designed to formally witness behavioral equivalence.
We use the term \define{behavior} as a synonym for the phrase \emph{bisimilarity (equivalence) class}.

\section{The final \texorpdfstring{\acro{GKAT}}{GKAT}-automaton}\label{sec:the_final_gkat-automaton}

One way of assigning semantics to \GKAT expressions is to find a sufficiently large \acro{GKAT}-automaton \(\Z\) that contains the behavior of every other \acro{GKAT}-automaton. In this section, we provide a concrete explicit description of such a ``semantic''  \acro{GKAT}-automaton---this is a crucial step towards being able to devise a completeness proof.

Concretely, \(\Z\) represents the behavior of a state as a tree that holds information about acceptance, rejection, and transitions to other states (which are subtrees).
Essentially, this tree is an unfolding of the transition graph from that state.

We describe these trees using partial functions.
Let us write \(A^+\) for the set of all non-empty words consisting of atoms.
The state space \(Z\) of \(\coalg Z\) is the set of all partial functions \(t : A^+ \rightharpoonup 2 + \Sigma\) with \(A \subseteq \dom(t)\), such that the following hold for all \(a \in A\) and \(x \in A^+\).
\begin{mathpar}
	\prftree{w \in \dom(t)}{t(w) \in \Sigma}{wa \in \dom(t)}
	\and
	\prftree{w \in \dom(t)}{t(w) \in 2}{wx \nin \dom(t)}
\end{mathpar}
The transition structure of \(\coalg Z\) is defined by the inferences
\begin{mathpar}
	\prftree{t(a) = 0}{t \downarrow a}
	\and
	\prftree{t(a) = 1}{t \Rightarrow a}
	\and
	\prftree{t(a) = p \in \Sigma}{t \trans{a|p} \lambda w. t(aw)}
\end{mathpar}
When \(t(w) \in \Sigma\), we will write \(\partial_w t\) for \(\lambda u.t(wu)\).
We can think of \(t \in Z\) as a tree where the root has leaves for atoms \(a \in A\) with \(t(a) = 1\), and a subtree for every \(a \in A\) with \(t(a) \in \Sigma\).

\begin{remark}
	Trees correspond to \emph{deterministic} (possibly \emph{infinite}) \emph{guarded languages}~\cite{gkat,kozentseng2008}.
	More precisely, every tree can be identified with a language \(L \subseteq {(A \cdot \Sigma)}^* \cdot A \cup {(A \cdot \Sigma)}^\omega\) satisfying (i) if \(wap\sigma, waq\sigma' \in L\), then \(p = q\); and (ii) if \(wa \in L\), then \(wap\sigma \nin L\) for any \(p\sigma\).
	% ; for any \(w \in {(A \cdot \Sigma)}^*\),
	% \(\sigma,\sigma' \in {(A \cdot \Sigma)}^* \cdot A \cup {(A \cdot \Sigma)}^\omega\),
	% \(a \in A\), and \(p,q \in \Sigma\).
	% \(wapw' \in L\) and \(waqw'' \in L\) implies \(p = q\), and either \(wa \nin L\) or \(wapw' \nin L\) for any \(a \in A\), \(p,q \in \Sigma\).
	We forgo a description in terms of guarded languages in favor of trees because these trees have the constraint about determinism built in.
\end{remark}

A \define{node} of \(t\) is a word $w\in A^*$ such that either \(w = \epsilon\) (the empty word), or \(w \in \dom(t)\) and \(t(w) \in \Sigma\).
We write \(\Node(t)\) for the set of nodes of \(t\).
A \define{subtree} of \(t\) is a tree \(t'\) such that \(t' = \partial_w t\) for some \(w \in \Node(t)\).
A \define{leaf} of \(t\) is a word \(w \in \dom(t)\) such that \(t(w) \in 2\).

Next, we specialize \cref{lem:concrete bisimilarity} to \(\Z\) (c.f.~\cite[Theorem~3.1]{RUTTEN20031}).

\begin{restatable}{lemma}{restatetreeconcretebisimilarity}\label{lem:tree concrete bisimilarity}
	\(R \subseteq Z \times Z\) is a bisimulation on \(\coalg Z\) iff for any \((t,s) \in R\) and \(a \in A\),
        (1)~\(t(a) = s(a)\); and
		(2)~if either \(\partial_a t\) or \(\partial_a s\) is defined, then both are defined and \((\partial_a t, \partial_a s) \in R\).
\end{restatable}

We can now prove that bisimilar trees in \(Z\) coincide.

\begin{restatable}[Coinduction]{lemma}{restatezissimple}\label{lem:Z is simple}
	If \(s,t \in Z\) are bisimilar, then \(s = t\).
\end{restatable}

Thus, to show that two trees are equal, it suffices to demonstrate a bisimulation that relates them.
% , which is to say, a relation satisfying the conditions of \cref{lem:tree concrete bisimilarity}.
This proof method is called \define{coinduction}.
We can also use \cref{lem:tree concrete bisimilarity} to define algebraic operations on \(Z\), and such definitions are said to be \define{coinductive}.
Many of the results in the sequel are argued using coinduction, and many of the constructions are coinductive.
With this in mind, we are now ready to prove that \(\Z\) contains every behavior that can be represented by a \acro{GKAT}-automaton, as follows.

\begin{restatable}{theorem}{restateZisthefinalcoalgebra}\label{thm:Z is the final coalgebra}
	\(\coalg Z\) is the final \acro{GKAT}-automaton. In other words, for every \acro{GKAT}-automaton \(\X\), there exists a unique \acro{GKAT}-automaton homomorphism \(!_\X\) from \(\X\) to \(\Z\).
\end{restatable}

Given a \acro{GKAT}-automaton \(\X\), the unique map \(!_\X\) assigns a tree from \(Z\) to each of its states.
In particular, recalling that the syntactic \GKAT-automaton \(\coalg{E}\) has \(\Exp\) as its set of states, \(!_{\coalg E}\) is a semantics of \acro{GKAT} programs in terms of trees.
The following lemma states that bisimulation is sound and complete with respect to this semantics.

\begin{restatable}{lemma}{restatefinalbisimilarity}%
\label{lem:final bisimilarity}
	States \(x\) and \(x'\) of a \acro{GKAT}-automaton \(\X\) are bisimilar iff \({!_\X(x)} = {!_\X(x')}\).
\end{restatable}

% *** *** SECTION *** *** %
\section{Trees form an algebra}\label{sec:i_g_as_an_algebra}

So far, we have seen that the behavior of a \acro{GKAT}-program is naturally interpreted as a certain kind of tree, and that each such tree is the state of the final \acro{GKAT}-automaton \(\coalg Z\).
In this section, we show that the trees in \(Z\) can themselves be manipulated and combined using the programming constructs of \acro{GKAT}.
These operations satisfy all of the axioms that build \(\equiv_0\), but fail the \emph{early-termination axiom} S3.
This gives rise to an inductive semantics of \acro{GKAT}-programs \(\sem{-}: \Exp \to Z\) that is sound w.r.t.~\(\equiv_0\).
As a matter of fact, we will see that \(\sem{-}\) coincides with the unique \acro{GKAT}-automaton homomorphism \(!_\coalg{E}: \Exp \to Z\).

We begin by interpreting the tests.
Given \(b \subseteq A\), we define \(\sem{b}\) as the characteristic function of \(b\) as a subset of \(A^+\), i.e., \(\sem{b}(a) = 1\) if \(a \in b\), and \(\sem{b}(a) = 0\) otherwise.

On the other hand, primitive action symbols denote programs that perform an action in one step and then terminate successfully in the next.
For \(p \in \Sigma\), this behavior is described by the unique tree \(\sem{p}\) such that \(\sem{p}(a) = p\) and \(\partial_a \sem{p} = \sem1\) for any \(a \in A\).
When context can disambiguate, we write \(b\) in place of \(\sem b\) and \(p\) in place of \(\sem p\).

Each operation is defined using a \define{behavioral differential equation (BDE)} consisting of a set of \define{initial conditions}
\(
	t(a) = \xi_a \in 2 + \Sigma
\)
indexed by \(a \in A\) and a set of \define{step equations}
\(
	\partial_a t = s_a
\)
indexed by the \(a \in A\) with \(t(a) \in \Sigma\).
This is possible because every BDE describes a unique automaton, which (by \cref{thm:Z is the final coalgebra}) has a unique interpretation in \(Z\)~\cite{RUTTEN20031}.
Each BDE below can be read more or less directly from \cref{tab:transition_structure_of_coalg_exp}.

The first operation that we interpret in \(Z\) is sequential composition.
For any \(s,t \in Z\), the tree \(s \cdot t\) models sequential composition of programs by replacing each non-zero leaf of \(s\) by the nodal subtree of \(t\) given by the corresponding atomic test.
This can formally be defined as the unique operation satisfying the following behavioral differential equation.

{\vspace{-1em}\small\begin{mathpar}
	(s \cdot t)(a) = \begin{cases}
		t(a) &\text{if \(s(a) = 1\)},\\
		s(a) &\text{otherwise}
	\end{cases} \and \qquad
	\partial_a (s \cdot t) = \begin{cases}
		\partial_a t &\text{if \(s(a) = 1\)},\\
		\partial_a s \cdot t &\text{otherwise.}
	\end{cases}
\end{mathpar}}%
Here, \(\partial_a s \cdot t = (\partial_a s) \cdot t\).
Using this operation, we define \(\sem{e \cdot f} = \sem{e} \cdot \sem{f}\).

To interpret the guarded union operation, define \(+_b\) to be the unique operation such that

{\vspace{-0.8em}\small\begin{mathpar}
	(s +_b t)(a) = \begin{cases}
		s(a) &\text{if \(a \in b\)},\\
		t(a) &\text{otherwise}
	\end{cases} \and
	\partial_a(s +_b t) = \begin{cases}
		\partial_a s &\text{if \(a \in b\)},\\
		\partial_a t &\text{otherwise.}
	\end{cases}
\end{mathpar}}%
As before, we define \(\sem{e +_b f} = \sem{e} +_b \sem{f}\).

Finally, we interpret the guarded exponential operation.
Following \cref{tab:transition_structure_of_coalg_exp}, \(t^{(b)}\) can be defined as the unique tree satisfying

{\vspace{-1em}\small\begin{mathpar}
	t^{(b)}(a) = \begin{cases}
		1 &\text{if \(a \nin b\)},\\
		t(a) &\text{if \(a \in b\) and  \(t(a) \in \Sigma\)},\\
		0 &\text{otherwise.}
	\end{cases}
	\and
	\partial_a (t^{(b)}) = \partial_a t \cdot t^{(b)}
\end{mathpar}}%
Similar to the other operators, we set \(\sem{e^{(b)}} = \sem{e}^{(b)}\).
This completes our definition of the algebraic homomorphism \(\sem- : \Exp\to Z\).

As it happens, \(\sem{-}\) is also a \GKAT automaton homomorphism from \(\coalg{E}\) to \(\Z\).
By uniqueness of such homomorphisms (\cref{thm:Z is the final coalgebra}), we can conclude that \(\sem{-}\) and \(!_{\coalg E}\) are the same.

\begin{restatable}{proposition}{restateetaisabialgebra}\label{prop:eta is a bialgebra}
	For any \(e \in \Exp\), \(\sem{e} =\ !_{\coalg E}(e)\).
\end{restatable}

This allows us to treat the algebraic and coalgebraic semantics as synonymous.
Using \cref{lem:final bisimilarity}, we can then show soundness w.r.t.~\(\equiv_0\) by arguing that \(\equiv_0\) is a bisimulation on \(\coalg E\).

\begin{restatable}{theorem}{restateZsatisfiesGKATminus}\label{thm:Z satisfies GKAT^-}
	The semantics \(\sem{-}\) is sound w.r.t.~\(\equiv_0\).
\end{restatable}

On the other hand, \(Z\) does not satisfy S3.
For instance, \(\sem{p \cdot 0} \neq \sem{0}\) for any \(p \in \Sigma\).
We will adapt the model to overcome this in \cref{sec:a_completeness_theorem_for_gkat}.

% *** *** SECTION *** *** %
\section{Well-nested automata and nested behavior}\label{sec:well_nested_coalgebras}

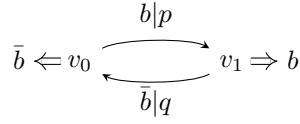
\begin{figure}[!t]
\centering
\begin{tikzpicture}[
	->,
	> = stealth,
	node distance = 2cm,
	every node/.style = {
		thick,
		minimum height=1.7em,
		text height=0.6em,
	},
]
	\node at (2, 0) (0) {\(v_0\)};
	\node at (4, 0) (1) {\(v_1\)};
	\node[right=3mm of 1] (2) {\(b\)};
	\node[left=3mm of 0] (3) {\(\bar{b}\)};

	\path (0) edge[above, bend left,looseness=0.5,->] node{\(b | p\)} (1);
	\path (1) edge[below, bend left,looseness=0.5,->] node{\(\bar b | q\)} (0);
	\path ($(1) + (2.5mm,0)$) edge[double,double distance=2pt,-implies] ($(2) + (-2mm,0)$);
	\path ($(0) + (-2.5mm,0)$) edge[double,double distance=2pt,-implies] ($(3) + (2mm,0)$);
\end{tikzpicture}
\caption{A \acro{GKAT}-automaton without \GKAT behaviors.}%
\label{fig:a non gkat}
\end{figure}

Not all behaviors expressible in terms of finite \GKAT-automata occur in \(\coalg E\).
For example, the two-state automaton in \cref{fig:a non gkat} fails to exhibit any behavior of the form \(\sem e\), with \(e \in \Exp\), when \(b, \bar b \neq 0\).
This is proven in \cref{appendix:well-nested automata}.
where we show that no branch of a \acro{GKAT} behavior can accept both \(b\) and \(\bar b\) infinitely often.
For another example, see~\cite{kozentseng2008}, where a particular three-state automaton is shown to exhibit no \acro{GKAT} behavior.

Intuitively, both of the examples above fail to exhibit the behaviors of \acro{GKAT} programs because \GKAT lacks a \textbf{goto}-statement that allows control to transfer to an arbitrary position in the program; instead, \GKAT automata corresponding to \GKAT expressions are structured by branches and loops.
The question then arises: can we characterize the ``shapes'' of automata whose behavior is \textbf{goto}-free, i.e., described by a \GKAT expression?

In~\cite{gkat}, the authors proposed the class of \emph{well-nested} \GKAT automata, consisting of automata built inductively by applying a series of operations designed to mimic the structural effects of loops.
It was shown that the behavior of every \GKAT expression can be described by some well-nested automaton.
Moreover, they proved that the class of well-nested automata constitutes a sufficient condition: the behavior of a well-nested \GKAT automaton is described by a \GKAT expression.
Whether this condition is also \emph{necessary}, i.e., whether every automaton with behavior corresponding to a \GKAT expression is well-nested, was left open.

Thus, a positive answer to the latter question amounts to showing that every \GKAT automaton whose behavior is the same as a well-nested \GKAT automaton is itself well-nested.
Such a class of automata closed under behavioral equivalence is known as a \define{covariety}.
Covarieties have desirable structural properties.
In particular, they are closed under homomorphic images~\cite{Rutten2000,gumm2000elementsofthegeneraltheoryofcoalgebras,adamekporst2003}.
Unfortunately, well-nested automata do not satisfy this property: we have found a well-nested automaton whose homomorphic image is not well-nested, depicted in~\cref{fig:_cat_wn}.
In other words, there exists a non-well-nested automaton whose behavior is still described by a \GKAT expression.
This also closes the door on a simpler approach to completeness described in~\cite{gkat}.

\smallskip

Thus, well-nested automata do not constitute a characterization of the \GKAT automata that correspond to \GKAT expressions.
To obtain such a characterization, we take a slightly different approach: rather than describing shapes of these automata, we describe the shapes of the trees that they denote.
We refer to a set of trees \(\coeq{U} \subseteq Z\) as a \define{coequation}, and treat it as a predicate: a \GKAT-automaton \(\X\) \define{satisfies} \(\coeq{U}\), written \(\X \models \coeq{U}\), if every behavior present in \(\X\) appears in \(\coeq{U}\) --- in other words, if \(!_\X\) factors through \(\coeq{U}\).
We write \(\Cov(\coeq U)\) to denote the class of all \acro{GKAT}-automata that satisfy \(\coeq U\).
It is easily shown that \(\Cov(\coeq{U})\) is a covariety.

The coequation that we give to describe the covariety of automata whose behavior corresponds to a \GKAT expression is driven by the intuition behind well-nested automata: the trees in this coequation are built using compositions that enforce \textbf{while}-like behavior, and do not permit the construction of \textbf{goto}-like behavior.
To this end, we need to define a new \emph{continuation} operation, as follows.
Given \(s,t \in Z\), the \define{continuation} \(s \rhd t\) of \(s\) along \(t\) is the unique tree satisfying the behavioral differential equation

{\vspace{-1em}\small\begin{mathpar} % HACK: if this is part of the previous paragraph, strange things happen.
	(s \rhd t) (a) = \begin{cases}
		t(a) &\text{if \(s(a) = 1\)},\\
		s(a) &\text{otherwise}
	\end{cases}
	\and
	\partial_a (s \rhd t) = \begin{cases}
		\partial_a t \rhd t &\text{if \(s(a) = 1\)},\\
		\partial_a s \rhd t &\text{otherwise.}
	\end{cases}
\end{mathpar}}%
Intuitively, \(s \rhd t\) is the tree that attaches infinitely many copies of \(t\) to \(s\).
This operation can be thought of as the dual to Kleene's original \(*\)-operation~\cite{Kleene1951}, which loops on its first argument some number of times before continuing in the second.

\begin{definition}
	The \define{nesting} coequation \(\coeq W\) is the smallest subset of \(Z\) containing the \define{discrete} coequation	\(\coeq D := \{\sem{b} \mid b \subseteq A\}\) and closed under the \define{nesting} rules below:
	{\small\begin{mathpar}
		    \infer{%
		        t, s \in \coeq{W}
		    }{%
		        t \cdot s \in \coeq{W}
		    }
		    \and
		    \infer{%
		        (\forall a \in A)\  t(a) \in \Sigma \implies \partial_a t \in \coeq W
		    }{%
		        t \in \coeq{W}
		    }
		    \and
		    \infer{%
		        t, s \in \coeq{W}
		    }{%
		        t \rhd s \in \coeq{W}
		    }
		\end{mathpar}}%
\end{definition}
The first and third nesting rules say that \(\coeq W\) is closed under composition and continuation; the second rule says that integrals over nested trees are nested.

\smallskip
It is not too hard to see that \(\coeq W\) is a subautomaton of \(\coalg Z\).
In other words, if \(t \in\coeq W\), then the derivatives of \(t\) are in \(\coeq W\) as well.
In fact, \(\coeq W\) is a subalgebra of \(Z\) in that it is closed under the operations of \acro{GKAT}.
This can be seen from the following observations: first, \(\partial_a p = 1\) for all \(a \in A\), so \(p \in \coeq W\) for any \(p \in \Sigma\) by the second nesting rule.
Second, \(\coeq W\) is closed under sequential composition by definition.
Third, if \(s,t \in \coeq W\) and \(b \subseteq A\), then every derivative of \(s +_b t\) is either a derivative of \(s\) or a derivative of \(t\).
Lastly, closure under the guarded exponential is a consequence of the identity

{\vspace{-0.8em}\small
\begin{mathpar}\textstyle
	t^{(b)} = 1 \rhd (\tilde{t} +_b 1),
	\qquad\text{where}\qquad
	\tilde t := \bigplus_{t \trans{a|p_a} t_a} p_a \cdot t_a.
\end{mathpar}}%
This identity can be shown to hold for all \(t \in Z\) and \(b \subseteq A\) using a coinductive argument.
It follows that the nesting coequation contains the image of \(\sem-\).
A similar argument can be used to establish the reverse containment as well, which leads to the following.

\begin{restatable}{proposition}{restateexistence}\label{prop:existence}
	\(\coeq{W}\) is the set of \acro{GKAT} program behaviors, i.e,
	\(
		\coeq W = \{\sem{e} \mid e \in \Exp\}.
	\)
\end{restatable}

\cref{prop:existence} characterizes \(\coeq{W}\) as the the set of behavioral patterns exhibited by \acro{GKAT} expressions: the states of a \acro{GKAT}-automaton \(\X\) behave like \acro{GKAT} programs if and only if \(\X\) satisfies \(\coeq W\), or, in other words, if \(\X\) can be found in the covariety \(\Cov(\coeq{W})\).
Since every well-nested automaton has the behavior of some \GKAT expression~\cite{gkat}, it must satisfy \(\coeq{W}\).
\begin{restatable}{proposition}{restatewisnecessary}%
	\label{prop:W is necessary}
	Well-nested \acro{GKAT}-automata satisfy  the nesting coequation.
\end{restatable}
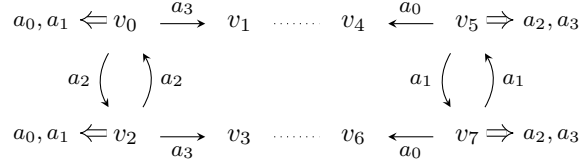
\begin{figure}[!t]%
	\centering
	\begin{tikzpicture}[
	->,
	> = stealth,
	node distance = 15mm,
	every state/.style = {
		thick,
		inner sep=0,
		fill=black!0.0,
		draw=black!0.0
	}
]
		\node[state] (0) {\(v_0\)};
		\node[state, right of=0] (1) {\(v_1\)};
		\node[state, below of=0] (2) {\(v_2\)};
		\node[state, right of=2] (3) {\(v_3\)};
		\node[state, right of=1] (4) {\(v_4\)};
		\node[state, right of=4] (5) {\(v_5\)};
		\node[state, below of=4] (6) {\(v_6\)};
		\node[state, right of=6] (7) {\(v_7\)};

		\node[state, left=2mm of 0] (a0) {\footnotesize\(a_0,a_1\)};
		\node[state, left=2mm of 2] (a2) {\footnotesize\(a_0,a_1\)};
		\node[state, right=2mm of 5] (a5) {\footnotesize\(a_2,a_3\)};
		\node[state, right=2mm of 7] (a7) {\footnotesize\(a_2,a_3\)};

		\draw (0) edge[above] node{\footnotesize\(a_3\)} (1)
			  (2) edge[below] node{\footnotesize \(a_3\)} (3)
			  (0) edge[left, bend right] node{\footnotesize \(a_2\)} (2)
			  (2) edge[right, bend right] node{\footnotesize\(a_2\)} (0)
			  (5) edge[above] node{\footnotesize\(a_0\)} (4)
			  (7) edge[below] node{\footnotesize\(a_0\)} (6)
			  (5) edge[left, bend right] node{\footnotesize\(a_1\)} (7)
			  (7) edge[right, bend right] node{\footnotesize\(a_1\)} (5);

		\draw (1) edge[-, dotted, above] node{} (4)
			  (3) edge[-, dotted, above] node{} (6);

		\path ($(0) + (-2.5mm,0.5mm)$) edge[double,double distance=2pt,-implies] ($(a0.east) + (0.5mm,0.5mm)$)
			  ($(2) + (-2.5mm,0.5mm)$) edge[double,double distance=2pt,-implies] ($(a2.east) + (0.5mm,0.5mm)$)
			  ($(5) + (2.5mm,0.5mm)$) edge[double,double distance=2pt,-implies] ($(a5.west) + (-0.5mm,0.5mm)$)
			  ($(7) + (2.5mm,0.5mm)$) edge[double,double distance=2pt,-implies] ($(a7.west) + (-0.5mm,0.5mm)$);
	\end{tikzpicture}
	\vspace*{-.25cm}
	\caption{As depicted, this automaton is well-nested. However, identifying \(v_1\) with \(v_4\), and \(v_3\) with \(v_6\), we obtain an automaton that is not well-nested.}\label{fig:_cat_wn}
	\vspace*{-.25cm}
\end{figure}

\section{Completeness}\label{sec:completeness}

This section contains two completeness theorems for \GKAT\@. As in~\cite{gkat}, we need to assume that W3 is generalized to arbitrary (linear) systems of equations. This \emph{uniqueness axiom}, discussed in \cref{sub:uniqueness_of_solutions_for_salomaa_systems}, will allow us to prove that the semantics \(\sem-\) from \cref{sec:i_g_as_an_algebra} is free with respect to \(\equiv_0\)---that is, \(\sem e = \sem f\) implies \(e \equiv_0 f\)---in \cref{sec:a_completeness_theorem_for_gkat-}. This will then provide an alternative route to completeness for \GKAT in \cref{sec:a_completeness_theorem_for_gkat}.

\subsection{Uniqueness of solutions for Salomaa systems}\label{sub:uniqueness_of_solutions_for_salomaa_systems}

In part, W3 from \cref{fig:GKAT axioms} ensures that the equation \(g \equiv e \cdot g +_b f\) with indeterminate \(g\) has at most one solution in \(\Exp/{\equiv_0}\) for any \(e,f \in \Exp\) under the condition that \(e\) denotes a productive program.
In fact, we could have stated the axiom this way from the beginning, as W1 provides the existence of a solution to this equation (even without the restriction on productivity).
As we will see, the uniqueness axiom makes a more general statement than W3 about \emph{systems} of equations with an arbitrary number of indeterminates.

\begin{definition}
A \define{system of (\(n\) left-affine) equations} is a sequence of \(n\) equations of the form
\(	x_i = e_{i1}\cdot x_1 +_{b_{i1}} \cdots +_{b_{i(n-1)}} e_{in}\cdot x_n +_{b_in} c_{i}\), % chktex 25
indexed by \(i \le n\), such that
(1) \(x_i\) is an indeterminate variable;
(2) \({(b_{ij})}_{j \le n}\) is a sequence of \define{disjoint} Boolean expressions, i.e.~\(b_{ij} \wedge b_{ik} \equiv 0\) for any \(j \neq k\);
(3) \(c_i\) is a Boolean expression disjoint from \(b_{ij}\) for all \(j \le n\); and
(4) \(e_{ij}\) is a \acro{GKAT} expression for any \(j \le n\).

Given any congruence \(\gequiv\) satisfying the axioms of \(\equiv_0\), a \define{solution in \(\Exp/\gequiv\)} to such a system is an \(n\)-tuple of \acro{GKAT} expressions \({(g_i)}_{i \le n}\) such that the equivalence
\(
	g_i \gequiv e_{i1}\cdot g_1 +_{b_{i1}} \cdots +_{b_{i(n-1)}} e_{in}\cdot g_n +_{b_{in}} c_i
\) holds
for all \(i \le n\).
\end{definition}
For example, the equation in the premise of W3 is a system of one left-affine equation, and the conclusion prescribes a unique solution (in \(\Exp/{\equiv_0}\)) to the premise.
Every finite \acro{GKAT}-automaton \(\X\) gives rise to a system of equations with variables indexed by \(X = \{ x_i \mid i \leq n \}\) and coefficients indexed by the transition map, as follows:

{\vspace{-0.5em}\small
\begin{mathpar}
	e_{ij} = \bigplus_{x_i \trans{a|p_a} x_j} p_a
    \and
	c_{i} = \{a \in A \mid x_i \Rightarrow a\}
    \and
	b_{ij} = \{a \in A \mid x_i \trans{a|p} x_j\}.
\end{mathpar}}%
Solving this system of equations uncovers the \acro{GKAT}-constructs the automaton implements.

The uniqueness axiom states that certain systems of equations, like the one in the premise of W3, admit at most one solution.
Choosing which systems the axiom should apply to must be done carefully for the same reason that necessitates the side-condition on W3.
Crucially, we require that the system have \emph{productive coefficients}, i.e.~\(E(e_{ij}) \equiv 0\) for all \(i,j \le n\), to admit a unique solution.
As this condition is analogous to Salomaa's \emph{empty word property}~\cite{salomaa1966}, a system of equations with productive coefficients is called \define{Salomaa}~\cite{gkat}.
The \define{uniqueness axiom (for \(\gequiv\))} states that every Salomaa system of equations has at most one solution in \(\Exp/{\gequiv}\).
It is sound with respect to the semantics \(\sem-\) from \cref{sec:i_g_as_an_algebra}.

\begin{restatable}{theorem}{restatesolvinginZ}\label{thm:solving in Z}
	For any \(i, j \le n\), let \(s_{ij} \in Z\) satisfy \(s_{ij}(a) \neq 1\) for any \(a \in A\), \({(b_{ij})}_{j \neq n}\) be a sequence of disjoint Boolean expressions for any \(i \le n\), and \(c_i \subseteq A\) be disjoint from \(b_{ij}\) for each \(i \le n\).
	The system of equations
	\(
		x_i = s_{i1}\cdot t_1 +_{b_{i1}} \cdots +_{b_{i(n-1)}} s_{in}\cdot t_n +_{b_{in}} c_i,
	\)
	indexed by \(i \le n\) has a unique solution in \(Z^n\).
\end{restatable}

\subsection{Completeness with respect to \texorpdfstring{\(\equiv_0\)}{0-equivalence}}\label{sec:a_completeness_theorem_for_gkat-}

Next, we present a completeness theorem w.r.t.~\(\equiv_0\).
We have already seen that the behavior of a program takes the form of a tree, and that the programming constructs of \acro{GKAT} apply to trees in such a way that equivalence up to the axioms of \(\equiv_0\) is preserved (\cref{thm:Z satisfies GKAT^-}).
The completeness theorem in this section shows that up to \(\equiv_0\)-equivalence, \acro{GKAT} programs can be identified with the trees they denote.

\begin{restatable}[Completeness for \(\equiv_0\)]{theorem}{restatecompletenessforgkatminus}\label{cor:completeness for GKAT^-}
	Assume the uniqueness axiom for \(\equiv_0\) and let \mbox{\(e, f \in \Exp\).}
	If \(\sem e = \sem f\), then \(e \equiv_0 f\).
\end{restatable}
\begin{proof}[Proof sketch.]
Since \(\sem{e} = \sem{f}\), \(e\) and \(f\) are bisimilar as expressions.
This bisimulation gives rise to a Salomaa system of equations, which can be shown to admit both the derivatives of \(e\) and \(f\) as solutions.
By the unique solutions axiom, it then follows that \(e \equiv_0 f\).
\end{proof}

% *** *** SECTION *** *** %
\subsection{Completeness with respect to \texorpdfstring{\(\equiv\)}{equivalence}}\label{sec:a_completeness_theorem_for_gkat}

Having found a semantics that is sound and complete w.r.t.\ \(\equiv_0\), we proceed to extend this result to find a semantics that is sound and complete w.r.t.\ \(\equiv\).
Recall that the only difference between these equivalences was S3, which equates programs that fail eventually with programs that fail immediately.
To coarsen our semantics, we need an operation on labelled trees that forces early termination in case an accepting state cannot be reached.

\begin{definition}
We say \(t \in Z\) is \define{dead} when for all \(w \in \dom(t)\) it holds that \(t(w) \neq 1\).
The \define{normalization} operator is defined coinductively, as follows:
{\small \begin{mathpar}
t^\wedge(a) =
	\begin{cases}
	0 & t(a) \in \Sigma \wedge \text{\(\partial_a t\) is dead}, \\
	t(a) & \text{otherwise}
	\end{cases} \and
\partial_a (t^\wedge) = {(\partial_a t)}^\wedge.
\end{mathpar}}%
\end{definition}

\begin{example}
	Normalizing the tree \(\sem{p +_b p \cdot 0}\) prunes the branch corresponding to \(\bar b\), since it has no accepting leaves.
	This yields the tree \(\sem{b \cdot p}\). %pictured above on the right.
\end{example}

We can compose the normalization operator with the semantics \(\sem{-}\) to obtain a new semantics \(\sem{-}^\wedge\), which replaces dead subtrees with early termination.
Composing normalization with the earlier semantics of \GKAT, we obtain the \define{normalized semantics} \(\sem{-}^\wedge\).
This semantics is sound w.r.t.~\(\equiv\).

\begin{restatable}{proposition}{restatehatWsatisfiesGKAT}%
\label{cor:hat W satisfies GKAT}
	If \(e \equiv f\), then \(\sem{e}^\wedge = \sem{f}^\wedge\).
\end{restatable}

For the corresponding completeness property, we need a way of ``normalizing'' a given expression in \(\Exp\).
The following observation gives us a way to do this.

\begin{restatable}{lemma}{restateclosureundernormalization}%
\label{lem:closure under normalization}
\(\mathsf{W}\) is closed under normalization.
\end{restatable}

When \(e \in \Exp\), we have that \(\sem{e} \in \coeq W\).
Moreover, by the above, \(\sem{e}^\wedge \in \coeq W\), which means that there is an \(e' \in \Exp\) such that \(\sem{e'} = \sem{e}^\wedge\).
We write \(e^\wedge\) for this \define{normalized expression}.
As it turns out, we can derive the equivalence \(e^\wedge \equiv e\) from the uniqueness axiom for \(\equiv\).
This gives an alternative proof of the completeness result of~\cite{gkat} that highlights the role of coequational methods in reasoning about failure modes.

\begin{restatable}[{\cite{gkat}}]{corollary}{restatecompletenessforgkat}\label{cor:completeness for GKAT}
Assume the uniqueness axiom for \(\equiv\) and \(\equiv_0\).
If \(\sem{e}^\wedge = \sem{f}^\wedge\), then \(e \equiv f\).
\end{restatable}
\begin{proof}[Proof sketch]
If \(\sem{e}^\wedge = \sem{f}^\wedge\), then \(\sem{e^\wedge} = \sem{f^\wedge}\).
By completeness of \(\equiv_0\) w.r.t.\ \(\sem{-}\), we can then derive that \(e \equiv e^\wedge \equiv_0 f^\wedge \equiv f\), and since \(\equiv_0\) is contained in \(\equiv\), also \(e \equiv f\).
\end{proof}

By normalizing the trees in \(\coeq W\), we obtain the coequation \(\coeq W^\wedge = \{t^\wedge \mid t \in \coeq W\}\).
This coequation precisely characterizes \acro{GKAT} programs with forced early termination.
In particular, since \(\coeq{W}^\wedge \subseteq \coeq{W}\), neither state in \cref{fig:a non gkat} has a semantics described by \(\sem{e}^\wedge\) for some \(e \in \Exp\).

\section{Related work}

This paper builds on~\cite{gkat}, where \acro{GKAT} was proposed together with a language semantics based on guarded strings~\cite{Kaplan69} and an axiomatization closely related to Salomaa's axiomatization
of regular expressions based on unique fixpoints~\cite{salomaa1966}.
Note that the language of \emph{propositional while programs} from~\cite{kozentseng2008,K08b} is closely related to \acro{GKAT} in terms of semantics, although the compact syntax and axiomatization were only introduced in~\cite{gkat}.

Some \GKAT-automata have behavior that does not correspond to any \acro{GKAT} expression, such as the example in~\cite{kozentseng2008}.
The upshot is that the B\"{o}hm-Jacopini theorem~\cite{boehm-jacopini-1966,harel-1980}, which states that every deterministic flowchart corresponds to a \textsc{while} program, does not hold propositionally, i.e., when we abstract from the meaning of individual actions and tests~\cite{kozentseng2008}.

In contrast with~\cite{gkat,kozentseng2008}, our work provides a precise characterization of the behaviors denoted by \GKAT\ programs using trees.
In other words, we characterize the image of the semantic map inside the space of all behaviors.
This explicit characterization was essential for proving completeness of the full theory of \GKAT, including the early termination axiom.
\acro{KAT} equivalence without early termination has been investigated by Mamouras~\cite{mamouras-2017}.

Brzozowski derivatives~\cite{Brzozowski1964} appear in the completeness proof of \acro{KA}~\cite{kozen-2001,kozen-2017,jacobs-2006}.
We were more directly inspired by Silva's coalgebraic analogues of Brzozowski derivatives used in the context of completeness~\cite{silva2010kleene}.
Rutten~\cite{RUTTEN20031} and Pavlovic and Escardo~\cite{PavlovicEscardo98} document the connection between the differential calculus of analysis and coalgebraic derivatives.

Coequations have appeared in the coalgebra literature in a variety of contexts, e.g.~\cite{adamekporst2003,adamek05,10.1016/j.ic.2015.08.001,DBLP:conf/mpc/SalamancaBBCR15,DBLP:conf/cmcs/SalamancaBR16}, and notably in the proof of generalized Eilenberg theorems~\cite{DBLP:conf/mfcs/UrbatACM17,DBLP:journals/tocl/AdamekMMU19}. The use of coequations in completeness proofs is, as far as we are aware, new.

% *** *** SECTION *** *** %
\section{Discussion}
\acro{GKAT} was introduced in~\cite{kozentseng2008} under the name \emph{propositional while programs} and extensively studied in~\cite{gkat} as an algebraic framework to reason about simple imperative programs.
We presented a new perspective on the theory of \acro{GKAT}, which allowed us to isolate a fragment of the original axiomatization that captures the purely behavioral properties of \acro{GKAT} programs.
We solved an open problem from~\cite{gkat}, providing a proof that well-nested automata are not closed under homomorphisms, thereby making it unlikely that these automata can be used in a completeness proof that does not rely on uniqueness axioms.
Finally, we proved completeness for the full theory, respecting the early-termination property, in which programs that fail immediately are equated with programs that fail eventually.

There are several directions for future work that are worth investigating.
First, it was conjectured in~\cite{gkat} that the uniqueness axiom follows from
%its special case W3 in the presence of
the other axioms of \acro{GKAT}.
This remains open, but at the time of writing we think this conjecture might be false.
Secondly, the technique we use, based on coequations, can serve as basis for a general approach to completeness proofs.
We plan to investigate other difficult problems where our technique might apply.
Of particular interest is an open problem posed by Milner in~\cite{milner1984}, which consists of showing that a certain set of axioms are complete w.r.t.~bisimulation equivalence for regular expressions.
Recently, Grabmeyer and Fokkink~\cite{grabmeyerfokkink2020} provided a partial solution.
We believe our technique can simplify their proofs and shed further light on Milner's problem.

We have chosen to adopt the axiomatization from~\cite{gkat}, which can be described as a Salomaa-style axiomatization---the loop is a unique fixpoint satisfying a side condition on termination.
We would like to generalize the results of the present paper to an axiomatization in which the loop is a least fixpoint w.r.t.~an order.
The challenge is that there is no natural order in the language because the $+$ of Kleene Algebra has been replaced by $+_b$.
However, we hope to devise an order $\leq$ directly on expressions and extend the characterizations that we have to the new setting.
This new axiomatization would have the advantage of being algebraic (that is, sound under arbitrary substitution), which makes it more suitable for verification purposes as the number of models of the language would increase.

%%
%% Bibliography
%%

%% Please use bibtex,

\bibliography{citations}

\appendix
%!TEX root = ./icalp_submission_main.tex

% \renewcommand{\thesectiondis}[2]{\Alph{section}:}

% *** SECTION *** %
\section{Detailed proofs for \cref{sec:the_final_gkat-automaton}: The final \acro{GKAT}-automaton}

\restatetreeconcretebisimilarity*
\begin{proof}
	If \(R\) is a bisimulation containing \((t,s)\), then for any \(a \in A\),
	\[
		s(a)
			= \begin{cases}
				1 &\text{if \(s \Rightarrow a\)},\\
				p &\text{if \(s \trans{a|p} \partial_a s\)},\\
				0 &\text{if \(s \downarrow a\)}.
			\end{cases}
			= \begin{cases}
				1 &\text{if \(t \Rightarrow a\)},\\
				p &\text{if \(t \trans{a|p} \partial_a t\)},\\
				0 &\text{if \(t \downarrow a\)}.
			\end{cases}
			= t(a)
	\]
	Furthermore, if \(\partial_a s\) is defined, then \(s(a) \in \Sigma\) by definition of \(Z\).
	Since \(t(a) = s(a)\), it follows that \(\partial_a t\) must also be defined; similarly, if \(\partial_a t\) is defined, so is \(\partial_a s\).
	Now, if \(\partial_a s\) and \(\partial_a t\) are defined, then \(s \trans{a|s(a)} \partial_a s\) and \(t \trans{a|t(a)} \partial_a t\); hence, \((\partial_a s, \partial_a t) \in R\) since \(R\) is a bisimulation.

	Conversely, suppose every pair \((t,s) \in R\) satisfies (1) and (2) above.
	By (1), \(s \Rightarrow a \iff t \Rightarrow a\) as well as \(s \downarrow a \iff t \downarrow a\).
	Furthermore, since \(t \trans{a|p} \partial_a t\) if and only if \(t(a) = p\), by (1) we find that \(s \trans{a|p} \partial_a s\) if and only if \(t \trans{a|p} \partial_a t\).
	By (2), \((\partial_a s, \partial_a t) \in R\) and we are done.
\end{proof}

\restatezissimple*
\begin{proof}
	Let \(R\) be a bisimulation.
	We claim that, for all \(w \in A^+\) and \((s,t) \in R\), we have (a) \(w \in \dom(s)\) if and only if \(w \in \dom(t)\); and (b) if \(w \in \dom(t) \cap \dom(s)\) then \(s(w) = t(w)\).
	
	The proof proceeds by induction on \(w\).
	In the base, \(w = a\) with \(a \in A\), in which case the first claim holds by definition of \(Z\), and the second claim follows from \(s\) and \(t\) being bisimilar.
	
	For the inductive step, let \(w = aw'\) for \(a \in A\) and assume the claim holds for \(w'\).
	If \(t(a) = s(a) \in 2\), then \(w \not\in \dom(s)\) and \(w \not\in \dom(t)\) by definition of \(Z\), so both claims hold immediately.
	Otherwise, if \(t(a) = s(a) \in \Sigma\), then both \(\partial_a s\) and \(\partial_a t\) are defined, and \((\partial_a s, \partial_a t) \in R\).
	For the first claim, we can derive by induction that
	\[
		w \in \dom(s)
		\iff w' \in \dom(\partial_a s)
		\iff w' \in \dom(\partial_a t)
		\iff w \in \dom(t)
	\]
	For the second claim, we also derive by induction that
	\(
		s(w) = \partial_a s(w') = \partial_a t(w') = t(w).
		\qedhere
	\)
\end{proof}

For the sake of the next proof, it is helpful to note that the \acro{GKAT}-automaton homomophism conditions can be rephrased.
Consider a function \(\varphi : X \to Y\) between the state spaces of two \acro{GKAT}-automata \(\X\) and \(\Y\).
Then \(\varphi\) is a \acro{GKAT}-automaton homomorphism if and only if\[
	\delta^{\Y}(\varphi(x), a) = \begin{cases}
		(p, \varphi(x')) &\text{if \(\delta^{\X}(x, a) = (p, x') \in \Sigma \times X\)}.\\
		\delta^{\X}(x, a) &\text{otherwise}.
	\end{cases}
\]
In particular, if \(\varphi\) is a \acro{GKAT}-automaton homomorphism, then if either \(\partial_a \varphi(x)\) or \(\varphi(\partial_a x)\) is defined, both are defined and \(\partial_a \varphi(x) = \varphi(\partial_a x)\).

\restateZisthefinalcoalgebra*
\begin{proof}
	Let \(\X = (X, \delta)\).
	First, we inductively extend \(\delta\) to \(\delta^* \colon X \times A^+ \rightharpoonup 2 + \Sigma\), as follows:
	\[
		\delta^*(x, w) =
			\begin{cases}
			\delta(x, a) & w = a \in A \wedge \delta(x, a) \in 2 \\
			p & w = a \in A \wedge \delta(x, a) = (p, x') \\
			\delta^*(x', w') & w = aw' \wedge \delta(x, a) = (p, x') \\
			\text{undefined} & \text{otherwise}
			\end{cases}
	\]
	The desired \acro{GKAT}-automaton homomorphism is then
	\[
		!_{\X}(x) := \lambda w.\delta^*(x, w).
	\]
	A straightforward argument shows that \(!_\X\) is well-defined, that is, \(!_\X(x)\) is a tree for each \(x \in X\).
	To see the homomorphism condition, first observe that if \(\delta(x, a) \in 2\), then
	\[
		!_{\X}(x)(a) = (\lambda w.\delta^*(x, w))(a) = \delta^*(x,a) = \delta(x, a).
	\]
	Furthermore, if \(x \trans{a|p} x'\), then \(!_{\X}(x)(a) = \delta^*(x, a) = p\) and 
	\[
	\partial_a !_{\X}(x) = \partial_a (\lambda w.\delta^*(x, w)) = \lambda w.\delta^*(x, aw) = \lambda w.\delta^*(\delta(x, a), w) = \lambda w.\delta^*(x', w) = {!_{\X}(x')}.
	\]

	To see uniqueness, let \(\varphi : \X \to \coalg Z\) be any \acro{GKAT}-automaton homomorphism.
	We use \cref{lem:tree concrete bisimilarity} to argue that the relation
	\[
		\{(!_\V(x), \varphi(x)) \mid x \in V\}
	\]
	is a bisimulation.
	First and foremost, 
	\begin{align*}
		\varphi(x)(a) &= \begin{cases}
			0 &\text{if \(x \downarrow a\)},\\
			1 &\text{if \(x \Rightarrow a\)},\\
			p &\text{if \(x \trans{a|p} \partial_a x\)}
		\end{cases}
		= \begin{cases}
			0 &\text{if \(\delta(x,a) = 0\)},\\
			1 &\text{if \(\delta(x,a) = 1\)},\\
			p &\text{if \(\delta^*(x,a) = p\)}
		\end{cases}
		=\ !_\X(x)(a).
	\end{align*}
	For the step equations, observe that\[
		!_\X(x)(a) \in \Sigma \iff (\exists p \in \Sigma)\ x \trans{a|p} \partial_a x \iff \varphi(x)(a) \in \Sigma,
	\]
	as well as that \(!_{\X}(\partial_a x) = \partial_a (!_{\X}(x))\) and \(\varphi(\partial_a x) = \partial_a \varphi(x)\).
	Hence,\[
		(\partial_a (!_{\X}(x)), \partial_a \varphi(x)) = (!_{\X}(\partial_a x), \varphi(\partial_a x)) \in R.
	\]
	By \cref{lem:tree concrete bisimilarity}, \(R\) is a bisimulation.
\end{proof}

\restatefinalbisimilarity*

\begin{proof}
	To see sufficiency, note that the graph of a \GKAT-automaton homomorphism is a bisimulation by definition.
	It is easily shown that the converse of a bisimulation is a bisimulation, as is the (relational) composition of two bisimulations.
	Composing the graph of \(!_{\X}\) with its converse puts the pair \((x, x')\) in a bisimulation on \(\X\).

	For necessity, let \(\bisim\) be the set of pairs of bisimilar states of \(\X\), and note that it forms an equivalence relation.
	Observe that the quotient map \(q: X \to X/\bisim\) is a \GKAT-automaton homomorphism for a unique \GKAT-automaton structure \(\X/\bisim\) on \(X/\bisim\).
	Because the composition of \GKAT-automaton homomorphisms is again a \GKAT-automaton homomorphism, we have two \GKAT-automaton homomorphisms from \(\X\) to \(\Z\): the map \(!_\X\) as well as \({!_{\X/\bisim}} \circ q\).
	By \cref{thm:Z is the final coalgebra}, these are the same; since \(q(x) = q(x')\), we conclude that \(!_{\X}(x) = {!_{\X}(x')}\).
\end{proof}

%  subsection  %

% *** SECTION *** %
\section{Detailed proofs for \cref{sec:i_g_as_an_algebra}: Trees form an algebra}

\restateetaisabialgebra*

\begin{proof}
	It suffices to show that \(\sem{-}\) is a \acro{GKAT}-automaton homomorphism.
	This amounts to show that the following rules hold:
	
    \vspace{-1em}\begin{mathpar}
		\infer{e \downarrow a}{\sem{e}(a) = 0}
		\and
		\infer{e \Rightarrow a}{\sem{e}(a) = 1}
		\and
		\infer{e \trans{a|p} e'}{\sem e \trans{a|p} \sem{e'}}
	\end{mathpar}
	We do this by induction on the transition rules for \(e\).
	In the base, there are two cases.
	\begin{itemize}
		\item By definition, \(\sem{b}(a) = 0\) if and only if \(b \downarrow a\), and \(\sem{b}(a) = 1\) if and only if \(b \Rightarrow a\).
		Since \(b\) does not admit any transitions in \(\coalg{E}\), the last implication holds vacuously.
		
		\item We have that \(p \trans{a|p} 1\) for any \(a \in A\); by definition of \(\sem{p}\), we have \(\sem p (a) = p\) and \(\partial_a \sem p = 1\), and hence \(\sem p \trans{a|p} \sem 1\).
		Furthermore, \(p\) does not terminate (succesfully or unsuccesfully) in \(\coalg{E}\), so the first two rules hold vacuously.
	\end{itemize}

	\noindent
	In the inductive step, suppose the three inferences above hold for \(e\) and \(f\), and \(b \subseteq A\).
	\begin{itemize}
		\item
		If \(e +_b f \downarrow a\), then either \(a \in b\) and \(e \downarrow a\), or \(a \in \bar b\) and \(f \downarrow a\).
		In the first case, \(\sem{e +_b f}(a) = \sem{e}(a) = 0\), and in the second \(\sem{e +_bf}(a) = \sem{f}(a) = 0\).
		
		Furthermore, if \(e +_b f \Rightarrow a\), then either \(a \in b\) and \(e \Rightarrow a\), or \(a \in \bar b\) and \(f \Rightarrow a\).
		In the first case, \(\sem{e +_bf}(a) = \sem{e}(a) = 1\), and in the second \(\sem{e +_b f}(a) = \sem f(a) = 1\).

		Finally, if \(e +_b f \trans{a|p} g\), then either \(a \in b\) and \(e \trans{a|p} g\), or \(a \in \bar b\) and \(f \trans{a|p} g\). 
		In the first case, \(\sem{e+_bf}(a) = \sem{e}(a) = p\) and \(\partial_a \sem{e +_bf} = \partial_a(\sem e +_b \sem f) = \partial_a \sem e = \sem g\), and in the second, \(\sem{e +_b f}(a) = \sem{f}(a) = p\) and \(\partial_a \sem{e +_b f} = \partial_a \sem{f} = \sem g\).

		\item If \(e \cdot f \downarrow a\), then either \(e \downarrow a\), or \(e \Rightarrow a\) and \(f \downarrow a\).
		In the first case, \(\sem{e}(a) = 0\) and \(\sem{e \cdot f}(a) = \sem{e}(a) = 0\),
		and in the second, \(\sem{e \cdot f}(a) = \sem{e} \cdot \sem{f}(a) = \sem{f}(a) = 0\).
		
		Furthermore, if \(e \cdot f \Rightarrow a\), then \(e \Rightarrow a\) and \(f \Rightarrow a\).
		Thus, \(\sem{e \cdot f}(a) = \sem{f}(a) = 1\). 

		Finally, if \(e \cdot f \trans{a|p} g\), then either \(e \Rightarrow a\) and \(f \trans{a|p} g\), or \(e \trans{a|p} e'\) and \(g = e' \cdot f\).
		In the first case, \(\sem{e \cdot f}(a) = \sem{f}(a) = p\) and \[
			\partial_a \sem{e \cdot f} = \partial_a(\sem e \cdot \sem f) = \partial_a \sem f = \sem g,
		\]
		meaning \(\sem{e \cdot f} \trans{a|p} \sem{g}\), and in the second \(\sem{e \cdot f}(a) = \sem{e}(a) = p\), and \[
			\partial_a \sem{e \cdot f} = \partial_a \sem e \cdot \sem f = \sem {e'} \cdot \sem f = \sem g,
		\]
		thus showing that \(\sem{e \cdot f} \trans{a|p} \sem{g}\) again.

		\item If \(e^{(b)} \downarrow a\), then \(a \in b\) and either \(e \downarrow a\) or \(e \Rightarrow a\).
		In either case, \(\sem{e^{(b)}}(a) = \sem{e}^{(b)}(a) = 0\).
		
		Furthermore, if \(e^{(b)} \Rightarrow a\), then \(a \in \bar b\) and \(\sem{e^{(b)}} = \sem{e}^{(b)}(a) = 1\) by definition.

		Finally, if \(e^{(b)} \trans{a|p} g\), then \(a \in b\), \(e \trans{a|p} e'\), and \(g = e' \cdot e^{(b)}\).
		This means that \(\sem{e^{(b)}}(a) = \sem{e}^{(b)}(a) = \sem{e}(a) = p\) and\[
			\partial_a \left[\hspace{-0.25em}\left[{e}^{(b)}\right]\hspace{-0.25em}\right] = \partial_a \sem{e}^{(b)} = \partial_a \sem{e} \cdot \sem{e}^{(b)} = \sem{e'} \cdot \left[\hspace{-0.25em}\left[{e}^{(b)}\right]\hspace{-0.25em}\right] = \sem{g}.
			\qedhere
		\]
	\end{itemize}
\end{proof}

\restateZsatisfiesGKATminus*
\begin{proof}
	We should show that if \(e, f \in \Exp\) with \(e \equiv_0 f\), then \(\sem{e} = \sem{f}\).
	By \cref{prop:eta is a bialgebra} and \cref{lem:final bisimilarity}, it suffices to show that \(\equiv_0\) is a bisimulation on \(\coalg{E}\).
	We do this by induction on \(\equiv_0\).
	The proof is somewhat long, but completely straightforward in almost all cases.
	
	In the base, we have one case to consider for each of the axioms.
	For the guarded union axioms U1 through U5, reflexivity of \(\equiv_0\) means that it suffices to show that if \(e \equiv_0 f\) as a consequence of one of these axioms, we have for all \(a \in A\) that \(e \downarrow a\) if and only if \(f \downarrow a\), as well as \(e \Rightarrow a\) if and only if \(f \Rightarrow a\), and \(e \trans{a|p} g\) if and only if \(f \trans{a|p} g\). 
	\begin{enumerate}%[align=left]
		\item[(U1)]
			If \(e = f +_b f\), for some \(b \in \operatorname{BExp}\), suppose \(a \in b\); then \(e \downarrow a\) if and only if \(f \downarrow a\) by definition of the transition structure on expressions; similarly, \(e \Rightarrow a\) if and only if \(f \Rightarrow a\), and \(e \trans{a|p} g\) if and only if \(f \trans{a|p} g\).
			The case for \(a \not\in b\) is argued similarly.

		\item[(U2)]
			If \(e = g_0 +_b g_1\) and \(f = g_1 +_{\overline{b}} g_0\) for some \(g_0, g_1 \in \Exp\) and \(b \in \operatorname{BExp}\), then suppose \(a \in A\).
			We then have \(e \downarrow a\) if and only if \(g_0 \downarrow a\) if and only if \(f \downarrow a\), by definition of \(\coalg{E}\).
			By a similar argument, \(e \Rightarrow a\) if and only if \(f \Rightarrow a\) and \(e \trans{a|p} h\) if and only if \(f \trans{a|p} h\).
			The case where \(a \not\in b\) is argued similarly.
		\item[(U3)]
			If \(e = (g_0 +_b g_1) +_c g_2\) and \(f = g_0 +_{b \wedge c} (g_1 +_c g_2)\) where \(g_0, g_1, g_2 \in \Exp\) and \(b,c \in \operatorname{BExp}\), then there are three cases, based on \(a \in A\).

			\begin{itemize}
				\item
				First, if \(a \in c \wedge b\), then \(e \downarrow a\) precisely when \(g_0 \downarrow a\), which holds if and only if \(f \downarrow a\).
				By a similar argument \(e \Rightarrow a\) if and only if \(g_0 \Rightarrow a\) if and only if \(f \Rightarrow a\).
				Likewise, \(e \trans{a|p} h\) if and only if \(g_0 \trans{a|p} h\) if and only if \(f \trans{a|p} h\).

				\item
				Next, if \(a \in c \wedge \overline{b}\), note that the latter is equivalent to \(a \in \overline{b \wedge c} \wedge c\).
				A similar argument then shows the same properties as in the previous case, except with \(g_1\).

				\item
				Finally, if \(a \in \overline{c}\) then note that in particular \(a \not\in b \wedge c\).
				We again recover the same properties as in the two previous cases.
			\end{itemize}
		\item[(U4)] 
			If \(e = g_0 +_b g_1\) and \(f = b \cdot g_0 +_b g_1\) for some \(g_0,g_1\in\Exp\) and \(b \in \operatorname{BExp}\), then suppose \(a \in b\).
			In that case, \(e \downarrow a\) if and only if \(g_0 \downarrow a\), which holds precisely when \(b \cdot g_0 \downarrow a\), which is true if and only if \(f \downarrow a\).
			By a similar argument \(e \Rightarrow a\) if and only if \(f \Rightarrow a\) and \(e \trans{a|p} h\) if and only if \(f \trans{a|p} h\).
			The case where \(a \not\in b\) is covered by a similar argument.

		\item[(U5)]
			If \(e = (g_0 +_b g_1) \cdot g_2\) and \(f = g_0 \cdot g_2 +_b g_1 \cdot g_2\) for some \(g_0,g_1,g_2\in\Exp\) and \(b \in \operatorname{BExp}\), first suppose \(a \in b\).
			We can then derive as follows:
			\begin{align*}
			e \downarrow a
				&\iff g_0 +_b g_1 \downarrow a \vee [g_0 +_b g_1 \Rightarrow a \wedge g_2 \downarrow a] \\
				&\iff g_0 +_b \downarrow a \vee [g_0 \Rightarrow a \wedge g_2 \downarrow a] \\
				&\iff g_0 \cdot g_2 \downarrow a \\
				&\iff f \downarrow a
			\end{align*}
			Similarly, we can derive
			\begin{align*}
			e \Rightarrow a
				&\iff g_0 +_b g_1 \Rightarrow a \wedge g_2 \Rightarrow a \\
				&\iff g_0 \Rightarrow a \wedge g_2 \Rightarrow a \\
				&\iff f \Rightarrow a
			\end{align*}
			Finally, we have that
			\begin{align*}
			e \trans{a|p} h
				&\iff g_0 +_b g_1 \trans{a|p} h \vee [g_0 +_b g_1 \Rightarrow a \wedge g_2 \trans{a|p} h] \\
				&\iff g_0 \trans{a|p} h \vee [g_0 \Rightarrow a \wedge g_2 \trans{a|p} h] \\
				&\iff f \trans{a|p} h
			\end{align*}
			The case where \(a \not\in b\) is argued similarly.
	\end{enumerate}

	\noindent
	For the sequential composition axioms, we show the properties required of bisimulation.
	\begin{enumerate}%[align=left]
		\item[(S1)]
		If \(e = g_0 \cdot (g_1 \cdot g_2)\) and \(f = (g_0 \cdot g_1) \cdot g_2\), then we derive
		\begin{align*}
		e \downarrow a
			&\iff g_0 \downarrow a \vee [g_0 \Rightarrow a \wedge g_1 \cdot g_2 \downarrow a]\\
			&\iff g_0 \downarrow a \vee [g_0 \Rightarrow a \wedge [g_1 \downarrow a \vee [g_1 \Rightarrow a \wedge g_2 \downarrow a]]] \\
			&\iff g_0 \cdot g_1 \downarrow a \vee [g_0 \cdot g_1 \Rightarrow a \wedge g_2 \downarrow a] \\
			&\iff f \downarrow a
		\end{align*}

		Similarly, for succesful termination we can derive
		\begin{align*}
		e \Rightarrow a
			&\iff g_0 \Rightarrow a \wedge g_1 \cdot g_2 \Rightarrow a \\
			&\iff g_0 \Rightarrow a \wedge [g_1 \Rightarrow a \wedge g_2 \Rightarrow a] \\
			&\iff [g_0 \Rightarrow a \wedge g_1 \Rightarrow a] \wedge g_2 \Rightarrow a \\
			&\iff g_0 \cdot g_1 \Rightarrow a \wedge g_2 \Rightarrow a \\
			&\iff f \Rightarrow a
		\end{align*}

		Finally, if \(e \trans{a|p} h\), then there are two cases to consider.
		\begin{itemize}
			\item
			If \(h = h' \cdot (g_1 \cdot g_2)\) with \(g_0 \trans{a|p} h'\), then \(g_0 \cdot g_1 \trans{a|p} h' \cdot g_1\), and hence \(f \trans{a|p} (h' \cdot g_1) \cdot g_2\).
			Since \(h' \cdot (g_1 \cdot g_2) \equiv_0 (h' \cdot g_1) \cdot g_2\), we are done.
			\item
			If \(g_0 \Rightarrow a\) and \(g_1 \cdot g_2 \trans{a|p} h\), then it suffices to show that \(f \trans{a|p} h\).
			First, if \(h = h' \cdot g_2\) and \(g_1 \trans{a|p} h'\), then \(g_0 \cdot g_1 \trans{a|p} h'\), and hence \(f \trans{a|p} h' \cdot g_2 = h\).
			Second, if \(g_1 \Rightarrow a\) and \(g_2 \trans{a|p} h'\), then \(g_0 \cdot g_2 \Rightarrow a\), and hence \(f \trans{a|p} h\).
		\end{itemize}

		\item[(S2)]
		If \(e = 0 \cdot f\), then a straightforward argument shows that \(e \downarrow a\) for all \(a \in A\); since \(0 \downarrow a\) for all \(a \in A\), this completes the proof.
		
		\item[(S4)]
		If \(e = 1 \cdot f\), then a straightforward argument shows that \(e \downarrow a\) if and only if \(f \downarrow a\), as well as \(e \Rightarrow a\) if and only if \(f \Rightarrow a\), and \(e \trans{a|p} h\) if and only if \(f \trans{a|p} h\).
		As with the cases for the guarded union axioms, this suffices.

		\item[(S5)]
		If \(e = f \cdot 1\), then another straightforward argument shows that \(e \Downarrow a\) if and only if \(f \downarrow a\), as well as \(e \Rightarrow a\) if and only if \(f \Rightarrow a\).
		Furthermore, if \(e \trans{a|p} h\), then \(h = h' \cdot 1\) with \(f \trans{a|p} h'\).
		Since \(h' \cdot 1 \gequiv_0 h'\), this completes the proof for this case.
	\end{enumerate}

	\noindent
	The final cases to consider in the base are the first two loop axioms.
	\begin{enumerate}%[align=left]
		\item[(W1)]
		If \(e = g \cdot g^{(b)} +_b 1\) and \(f = g^{(b)}\) with \(g \in \Exp\) and \(b \in \operatorname{BExp}\), then we derive
		\[
			e \downarrow a
				\iff a \in b \wedge [g \downarrow a \vee [g \Rightarrow a \wedge g^{(b)} \downarrow a]]
				\iff f \downarrow a
		\]
		As far as succesful termination is concerned, we can derive
		\[
			e \Rightarrow a
				\iff [a \in b \wedge g \Rightarrow a \wedge g^{(b)} \Rightarrow a] \vee a \not\in b
				\iff f \Rightarrow a
		\]
		Finally, if \(e \trans{a|p} h\), then \(a \in b\) and \(h = g' \cdot g^{(b)}\) with \(g \trans{a|p} g'\).
		But in that case \(f \trans{a|p} h\) as well.
		Since \(\equiv_0\) is reflexive, this completes the proof.

		\item[(W2)]
		If \(e = {(c \cdot g)}^{(b)}\) and \(f = {(g +_c 1)}^{(b)}\) with \(g \in \Exp\) and \(b, c \in \operatorname{BExp}\), then derive
		\[
		e \downarrow a
			\iff a \in b \wedge [c \cdot g \downarrow a \vee c \cdot g \Rightarrow a] 
			\iff a \in b \wedge [g +_c 1 \downarrow a \vee g +_c 1 \Rightarrow a]
			\iff f \downarrow a
		\]
		Similarly, for succesful termination we derive
		\[
			e \Rightarrow a
				\iff a \not\in b
				\iff f \Rightarrow a
		\]
		Finally, if \(e \trans{a|p} h\), then \(h = h' \cdot e\) with \(c \cdot g \trans{a|p} h'\).
		Since \(c\) does not permit any transitions, this implies that \(a \in c\) and \(g \trans{a|p} h'\).
		From this, it follows that \(g +_c 1 \trans{a|p} h'\), and htus \(f \trans{a|p} h' \cdot f\).
		Since \(h' \cdot e \equiv_0 h' \cdot f\) by W2, we are done.
	\end{enumerate}

	\noindent
	The inductive cases for reflexivity, symmetry and transitivity of \(\equiv_0\) are completely straightforward, and follow from the fact that bisimilarity enjoys the same properties.
	
	To account for the fact that \(\equiv_0\) is a congruence, we treat the case for sequential composition, i.e., where \(e = e_0 \cdot e_1\) and \(f = f_0 \cdot f_1\) with \(e_0 \equiv_0 f_0\) and \(e_1 \equiv_0 f_1\); the other cases are similar.
	By induction, this tells us that \(e_0\) is bisimilar to \(f_0\), and \(e_1\) is bisimilar to \(f_1\).
	It is then not hard to show that \(e \downarrow a\) if and only if \(f \downarrow a\) as well as \(e \Rightarrow a\) if and only if \(f \Rightarrow a\).
	Furthermore, if \(e \trans{a|p} e'\), then either \(e' = e_0' \cdot e_1\) and \(e_0 \trans{a|p} e_0'\), or \(e_0 \Rightarrow a\) and \(e_1 \trans{a|p} e'\).
	In the former case, \(f_0 \trans{a|p} f_0'\) such that \(e_0 \equiv_0 f_0\), by induction.
	In that case \(f \trans{a|p} f_0' \cdot f_1\); since \(h = e_0' \cdot e_1 \equiv_0 f_0' \cdot f_1\), we are done.
	Otherwise, if \(e_0 \Rightarrow a\) and \(e_1 \trans{a|p} e'\), then by induction \(f_1 \trans{a|p} f'\) such that \(e' \equiv_0 f'\).
	Since furthermore \(f \trans{a|p} f'\) in this case, we are done.

	The only case where we need a new idea is for W3.
	Here, we know that \(e \equiv_0 f\) because \(f = g^{(b)} \cdot h\), with \(e \equiv_0 g \cdot e +_b h\) and \(E(g) \equiv_0 0\).
	A routine argument shows that \(e \downarrow a\) if and only if \(f \downarrow a\) as well as \(e \Rightarrow a\) if and only if \(f \Rightarrow a\).
	Next, if \(e \trans{a|p} e'\), then we know by applying the induction hypothesis to \(e \equiv_0 g \cdot e +_b h\) that \(g \cdot e +_b h \trans{a|p} e''\) with \(e' \equiv_0 e''\).
	This gives us two cases to consider.
	\begin{itemize}
		\item
		If \(a \in b\), then \(g \cdot e \trans{a|p} e''\).
		Now, note that if \(g \Rightarrow a\), then \(E(g) \Rightarrow a\) as well; since the latter would imply, by induction, that \(0 \Rightarrow a\), we can exclude it.
		This tells us that \(e'' = g' \cdot e\) with \(g \trans{a|p} g'\).
		In that case, \(f \trans{a|p} g' \cdot f\).
		Since \(e'' = g' \cdot e \equiv_0 g' \cdot f\), we are done.

		\item
		If \(a \not\in b\), then \(h \trans{a|p} e''\).
		In that case, \(g^{(b)} \Rightarrow a\), and hence \(f \trans{a|p} e''\).
		\qedhere
	\end{itemize}
\end{proof}

\section{Topological Structure of \(Z\)}\label{appendix:topological structure}

The space of trees \(Z\) has a rich structure that is useful in the proofs that follow.
In this appendix, we will show that we can equip \(Z\) with the compact metric \(d\), defined

\vspace{-1em}\begin{mathpar}
	d(s, t) = \max\left\{2^{-|w|} \mathrel{\Big|} \begin{array}{c}w \in \dom(s)\cap\dom(t) \\\text{ and }t(w) \neq s(w)\end{array}\right\},
\end{mathpar}
where \(\max\emptyset = 0\).

\begin{lemma}
	\((Z,d)\) is a metric space.
\end{lemma}

\begin{proof}
	Let \(s,t \in Z\).
	To show that \(d\) is a metric, we need to prove that \(s = t\) if and only if \(d(s,t) = 0\), and that \(d\) satisfies the triangle inequality.

	We begin by making the observation that, if \(w \in \dom(t) \setminus \dom(s)\), then \(d(s,t) > 2^{-|w|}\).
	Let \(w \in \dom(t)\setminus \dom(s)\).
	Since \(A \subseteq \dom(s) \cap \dom(t)\), there is a longest prefix \(w'a\) of \(w\) such that \(w'a \in \dom(s)\cap \dom(t)\).
	By assumption, \(t(w'a) \in \Sigma\), for otherwise \(w\) is a leaf of \(t\) and \(w=w'a\), contradicting the assumption that \(w \nin \dom(s)\).
	Moreover, \(s(w'a) \in 2\), for otherwise \(w'a\) would be a node of \(s\) and we could find a prefix \(w'au\) of \(w\), for some \(u \in A^+\), such that \(wau \in \dom(s) \cap \dom(t)\), contradicting the assumption that \(w'a\) is the longest prefix of \(w\) in \(\dom(s)\cap\dom(t)\).
	This means that \(t(w'a) \neq s(w'a)\), because \(\Sigma \cap 2 = \emptyset\).
	Hence, \(d(s,t) \ge 2^{-|w'a|} > 2^{-|w|}\).

	One consequence of this observation is that, if \(d(s,t) = 0\), then \(\dom(s) = \dom(t)\).
	Since this means that \(\dom(t) = \dom(s) \cap \dom(t) = \dom(s)\), \(d(s,t) = 0\) implies that \(s(w) = t(w)\) for any \(w\) where either is defined.
	Hence, \(s = t\).

	To see that \(d\) satisfies the triangle inequality, assume \(d(s,t) = 2^{-k}\).
	Then there is a word \(w \in \dom(s)\cap\dom(t)\) such that \(|w| = k\) and \(s(w) \neq t(w)\).
	Now consider a third tree, \(r \in Z\).
	It cannot be the case that both \(w \in \dom(r)\) with \(s(w) = r(w)\) and \(r(w) = t(w)\), so either \(w \in \dom(s)\setminus \dom(r)\), in which case \(d(s,r) > 2^{-k}\), or \(w \in \dom(r)\) with \(s(w) \neq r(w)\) or \(r(w) \neq t(w)\), meaning one of \(d(s,r)\) and \(d(r,t)\) is at least \(2^{-k}\).
	Whence,\[
		d(s,t) = 2^{-k} \le \max\{d(s,r), d(r,t)\} \le d(s,r) + d(r,t).
	\]
	This concludes the proof that \(d\) is a metric.
\end{proof}

Next, we argue that \((Z,d)\) is a complete metric space by showing something much stronger: \((Z,d)\) is compact.

\begin{lemma}
	\((Z,d)\) is a compact metric space.
\end{lemma}

\begin{proof}
	Let \((t_i)_{i >0}\) be an infinite sequence in \(Z\).
	To show that \(Z\) is compact, we need to exhibit a convergent subsequence of \((t_i)_{i > 0}\).
	This can be done as follows.

	Let \(\mathbf t^{(0)}=(t_i)_{i >0}\), and for any \(k \in \N\) let \(\mathbf t^{(k+1)}\) be a subsequence of \((t_{i}^{(k)})_{i > 0}\) satisfying \[
		(\forall i, j \in \N)(\forall w \in A^+)\ |w| \le k+1 \implies t_{i}^{(k)}(w) = t_{j}^{(k)}(w)
	\]
	Such a subsequence always exists, because there are finitely many partial functions \(\bigcup_{i=1}^{n}A^i \rightharpoonup 2 + \Sigma\), and hence there are infinitely many \(t_i^k\) that agree on all words of length at most \(k+1\).
	We claim that the subsequence \((t_i^{(i)})_{i > 0}\) of \((t_i)_{i >0}\) converges.

	The intuitive candidate for the limit of \((t_i^{(i)})_{i > 0}\) is given by the expression \(s = \lambda w.t_{|w|}^{(|w|)}(w)\).
	We need to show that this defines a tree in \(Z\).
	This can be done by induction on the domain rules for a tree in \(Z\).

	For the first domain rule, notice that if \(w \in \dom(s)\) and \(s(w) \in \Sigma\), then let \(n = |w|\) to find \(t_{n}^{(n)}(w) \in \Sigma\). 
	By construction, \(t_{n+1}^{(n+1)}(w) = t_n^{(n)}(w)\), putting \(t_{n+1}^{(n+1)}(w) \in \Sigma\).
	This means that for any \(a \in A\), \(wa \in \dom(t_{n+1}^{(n+1)})\).
	This puts \(wa \in \dom(s)\) for every \(a \in A\).

	For the second domain rule, let \(s(w) \in 2\).
	Where \(n = |w|\), \(t_n^{(n)}(w) = s(w)\), so \(t_n^{(n)}(w) \in 2\) as well.
	By construction, \(t_{n+k}^{(n+k)}(w) = t_n^{(n)}(w)\) for any \(k \ge 0\), putting \(t_{n+k}^{(n+k)}(w) \in 2\) for any \(k \ge 0\).
	If \(u \in A^+\) with \(|u| = k\), then \(wu \nin \dom(t_{n+k}^{(n+k)})\).
	Hence, \(wu \nin \dom(s)\).
	This concludes the argument showing that \(s \in Z\).

	For any \(n > 0\), and \(w \in \dom(s)\) with \(|w| \le n\), \(s(w) = t_{|w|}^{(|w|)}(w) = t_n^{(n)}(w)\).
	This means that \(s\) and \(t_n^{(n)}\) agree on all words of length at most \(n\), or equivalently \(d(s, t_{n}^{(n)}) \le 2^{-n}\).
	As \(n\) tends to \(\infty\), the subsequence \(t_n^{(n)}\) of \((t_i)_{i > 0}\) converges to \(s\).
	Hence, \(Z\) is compact. 
\end{proof}

Indeed, every compact metric space is also complete, for every incomplete metric space contains a sequence with no convergent subsequence (consider an arbitrary nonconvergent Cauchy sequence).
It should be noted, as well, that the completeness of \(Z\) does not depend on the finiteness of \(\Sigma\).
In fact, at the time of writing, the finiteness of \(\Sigma\) plays little to no role in the theory of \acro{GKAT} whatsoever.

% *** SECTION *** %
\section{Detailed proofs for \cref{sec:well_nested_coalgebras}: Well-nested automata and nested behaviour}\label{appendix:well-nested automata}

We begin this appendix by showing that our two-state automaton is not nested.
Define \(N(t) = \{a \in A \mid t(a) \in \Sigma\}\).

\begin{example}
	The automaton \(\X\) below is not nested if \(b,\bar b \neq 0\).
	
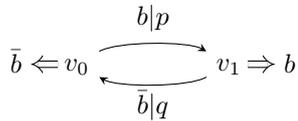
\[
	\begin{tikzpicture}[
		->,
		> = stealth,
		node distance = 2cm,
		every node/.style = {
			thick, 
			minimum height=1.7em,
			text height=0.6em,
		},
	]
		\node at (2, 0) (0) {\(v_0\)};
		\node at (4, 0) (1) {\(v_1\)};
		\node[right=3mm of 1] (2) {\(b\)};
		\node[left=3mm of 0] (3) {\(\bar{b}\)};	

		\path (0) edge[above, bend left,looseness=0.5,->] node{\(b | p\)} (1);
		\path (1) edge[below, bend left,looseness=0.5,->] node{\(\bar b | q\)} (0);
		\path ($(1) + (2.5mm,0)$) edge[double,double distance=2pt,-implies] ($(2) + (-2mm,0)$);
		\path ($(0) + (-2.5mm,0)$) edge[double,double distance=2pt,-implies] ($(3) + (2mm,0)$);
	\end{tikzpicture}
	\]
	This is a direct consequence of the following lemma.
\end{example}

\begin{lemma}\label{lem:even/odd nesting invariant}
	Let \(b \subset A\), \(t \in\coeq W\), and consider any infinite \emph{branch}\[
		B = \{\epsilon, a_1, a_1a_2, a_1a_2a_3, \dots\} \subseteq \Node(t)
	\]
	of \(t\).
	Then either
	
    \vspace{-1em}\begin{mathpar}
		|\{w \in B \mid E(\partial_w t) = b\}| < \omega
		\and\text{or}\and
		|\{w \in B \mid E(\partial_w t) = \bar b\}| < \omega.
	\end{mathpar}
	A branch with this property will be known as \define{finitely alternating}.
\end{lemma}

\begin{proof}
	By induction on the construction of \(t\).
	Since discrete trees do not have infinite branches, the base case is vaccuous.
	
	For the induction step, we assume that the lemma holds for any \(b \subset A\) and any infinite branch of \(r, s\) and the items of a sequence \(s_a\) indexed by \(A\).
	\begin{enumerate}
		\item[(\(+\))] Suppose \(\partial_a t = s_a\) for all \(a \in N(t)\), and consider a particular \(a \in N(t)\).
		If \(B\) is a branch of \(t\) including \(a\), then
		\(
			B = \{\epsilon\} \cup a B'
		\)
		for some branch \(B'\) of \(s_a\).
		Thus, since \(B'\) is finitely alternating by assumption, \(B\) must be as well.

		\item[(\(\cdot\))] Suppose \(t = r \cdot s\).
		Similarly, if \(B\) is an infinite branch of \(t\), then either \(B\) is an infinite branch of \(r\) or there is a word \(a_1\cdots a_n \in B\) such that 
		\begin{equation}\label{eq:B fin alt in cdot case}
			B = \{\epsilon, a_1, \dots, a_1\cdots a_n\} \cup (a_1\cdot a_{n-1})B'
		\end{equation}
		for some branch \(B'\) of \(s\) beginning with \(a_n\).
		Since there are only finitely many words of length at most \(n\) in \(B\),
		
        \vspace{-1em}\begin{mathpar}
			|\{w \in B \mid |w| \le n \text{ and } E(\partial_w t) = b\}| < \omega
			\and\text{and}\and
			|\{w \in B \mid |w| \le n \text{ and } E(\partial_w t) = \bar b\}| < \omega.
		\end{mathpar}
		Since \(B'\) is finitely alternating, it follows from \cref{eq:B fin alt in cdot case} that \(B\) must be as well.

		\item[(\(\rhd\))] Suppose \(t = r \rhd s\), and let \(B\) be an infinite branch of \(t\).
		Without loss of generality, we can assume that \(B\) is not a branch of \(r \cdot s^{\cdot n}\) for any \(n \in \N\) (by referring to the previous case otherwise).
		This means that, for some word \(wa \in B\) and \(n > 0\), \(w \in \Node(r \cdot s^{\cdot n})\) and \(r \cdot s^{\cdot n}(wa) = 1\).
		Simultaneously, however, \(wa \in \Node(r \rhd s)\), so it must be that \(s(a) \neq 1\).

		Assume for a contradiction that \(B\) infinitely alternates between accepting \(b\) and \(\bar b\), and without loss of generality assume that \(a \in b\).
		Since \(\{w \in B \mid E(\partial_w (r \rhd s)) = b\}\) is infinite, there is an \(m > n\) and a word \(aw' \in \Node(s^{\cdot (m-n)})\) such that \(waw' \in B\) and \(E(\partial_{waw'} (r \rhd s)) = b\).
		This means that \((r \rhd s)(waw'a) = 1\), as we assumed \(a \in b\), which is equivalent to saying that for any \(k \ge m\), \(r \cdot s^{\cdot k}(waw'a) = 1\).
		This contradicts the construction of \(t\), however, as we assumed \(s(a) \neq 1\) and therefore\[
			r \cdot s^{\cdot (m + 1)}(waw'a) = r \cdot s^{\cdot m} \cdot s(waw'a) = s(a) \neq 1.
		\]
		It follows that \(B\) must have been finitely alternating to begin with.
        \qedhere
	\end{enumerate}
\end{proof}

Let \(t \in Z\) and \(b \subseteq A\).
The observation that \(\coeq W\) is a subalgebra of \(Z\) rested on the the identity
\begin{mathpar}\textstyle
	t^{(b)} = 1 \rhd (\tilde{t} +_b 1),\ \text{ where }\ \tilde t = \bigplus_{t \trans{a|p_a} t_a} p_a \cdot t_a.
\end{mathpar}
This is established by showing that the relation
\(
	\{(s \cdot t^{(b)}, s \rhd (\tilde{t} +_b 1)) \mid s, t \in Z\} \cup \Delta_{Z}
\) 
is a bisimulation with \cref{lem:tree concrete bisimilarity}.
To this end, observe that 
\begin{align*}
	(s \cdot t^{(b)})(a) 
	&= \begin{cases}
		t^{(b)}(a) 	&\text{if \(s(a) = 1\)}\\
		s(a) 		&\text{otherwise}
	\end{cases}\\
	&= \begin{cases}
		1 			&\text{if \(a \nin b\) and \(s(a) = 1\)}\\
		t(a) 		&\text{if \(a \in b\), \(t(a) \in \Sigma\), and \(s(a) = 1\)}\\
		0 			&\text{if \(a \in b\), \(t(a) \in 2\), and \(s(a) = 1\)}\\
		s(a) 		&\text{otherwise}
	\end{cases}\\
	&= \begin{cases}
		1 			&\text{if \(a \nin b\) and \(s(a) = 1\)}\\
		\tilde t(a) &\text{if \(a \in b\) and \(s(a) = 1\)}\\
		s(a) 		&\text{otherwise}
	\end{cases}\\
	&= (s \rhd (\tilde{t} +_b 1))(a)
\end{align*}
This establishes (1) from \cref{lem:tree concrete bisimilarity}.
For (2), write 
\begin{align*}
	\partial_a (s \cdot t^{(b)})
	&= \begin{cases}
		\partial_a t \cdot t^{(b)} 		&\text{if \(s(a) = 1\), \(a \in b \wedge N(t)\)}\\
		\partial_a s \cdot t^{(b)} 		&\text{otherwise}
	\end{cases}\\
	\partial_a (s \rhd (\tilde t +_b 1))
	&= \begin{cases}
		\partial_a \tilde t \cdot (\tilde t +_b 1) 		&\text{if \(s(a) = 1\), \(a \in b \wedge N(t)\)}\\
		\partial_a s \cdot (\tilde t +_b 1) 			&\text{otherwise}
	\end{cases}\\
	&= \begin{cases}
		\partial_a t \cdot (\tilde t +_b 1) 	&\text{if \(s(a) = 1\), \(a \in b \wedge N(t)\)}\\
		\partial_a s \cdot (\tilde t +_b 1) 	&\text{otherwise}
	\end{cases}
\end{align*}
Each respective pair is a member of \(R\), so \(R\) is a bisimulation by \cref{lem:tree concrete bisimilarity}.

\restateexistence*
\begin{proof}
	We have already seen \(\coeq W \supseteq \img(\sem-)\).
	The reverse containment can be shown by induction on the nesting rules.

	By definition, \(\sem b \in \img(\sem-)\) for any \(b \subseteq A\).
	Furthermore, if \(\partial_a t = \sem{e_a}\) for all \(a \in N(t)\), then \[
		t = 1+_{E(t)} \left(\bigplus_{a \in N(t)}t(a) \cdot \sem{e_a}\right) 
		= \sem{1 +_{E(t)} \left(\bigplus_{a \in N(t)} t(a) \cdot e_a\right)}.
	\]
	If \(s = \sem{e}\) and \(t = \sem{f}\), then \(s \cdot t = \sem{e \cdot f}\) by definition.

	The continuation case can be seen from the following identity,
	\begin{equation}\label{eq:rhd and guarded exp}
		s \rhd t = s \cdot t^{(\overline{E(t)})}.
	\end{equation}
	If \(s = \sem e\) and \(t = \sem f\), then \[
		s \rhd t = \sem e \cdot \sem f^{(\overline{E(t)})} = \sem{e \cdot f^{(\overline{E(t)})}}.
	\]
	It now suffices to see \cref{eq:rhd and guarded exp}.
	This can be shown with a routine coinductive argument, establishing that\[
		R = \{(s \rhd t, s \cdot t^{(\overline{E(t)})}) \mid s, t \in Z\}
	\]
	is a bisimulation.
	Calculating, we see that both \((s \rhd t)(a)\) and \((s \cdot t^{(\overline{E(t)})})(a)\) are\[
		\begin{cases}
			t(a) &\text{if \(s(a) = 1\)},\\
			s(a) &\text{otherwise}.
		\end{cases}
	\]
	For the coinductive step, observe that
	\[
		\partial_a (s \cdot t^{(\overline{E(t)})}) = \begin{cases}
			\partial_a (t^{(\overline{E(t)})}) &\text{if \(s(a) = 1\)},\\
			\partial_a s \cdot t^{(\overline{E(t)})} &\text{otherwise.}
		\end{cases}
		= \begin{cases}
			\partial_a t \cdot t^{(\overline{E(t)})} &\text{if \(s(a) = 1\)},\\
			\partial_a s \cdot t^{(\overline{E(t)})} &\text{otherwise.}
		\end{cases}
	\]
	and
	\[
		\partial_a (s \rhd t) = \begin{cases}
			\partial_a t \rhd t &\text{if \(s(a) = 1\)},\\
			\partial_a s \rhd t &\text{otherwise.}
		\end{cases}
	\]
	The respective pairs are in \(R\), as desired.
	This establishes \cref{eq:rhd and guarded exp}.
\end{proof}

To formally define what it means to be \emph{well-nested}, we need the following automata-theoretic construction.
Given a \acro{GKAT}-automaton \(\X\), a subset \(U \subseteq X\), and a function \(h : A \to 2 + \Sigma\times X\), the \define{uniform continuation of \(h\) along \(U\)} is the automaton \(\X[U,h] = (X, \delta[U,h])\) obtained by setting\[
	\delta[U,h](x)(a) = \begin{cases}
		h(a) & \text{if \(x \Rightarrow a\) and \(x \in U\)},\\
		\delta(x)(a) & \text{otherwise}.
	\end{cases}
\]
% The following is a simple observation, but plays a key role in \cref{prop:W is necessary} below.
% 
% \begin{lemma}
% 	Let \(\X\) and \(\Y\) be \acro{GKAT}-automata, \(h \in G(X+Y)\).
% 	Then \(\Y\) is a subautomaton of \(\V = (\X + \Y)[X,h]\).
% \end{lemma}
A \acro{GKAT}-automaton \(\X\) is called \define{discrete} if it satisfies the discrete coequation, \(\coeq D\).
The class of \define{well-nested} \acro{GKAT}-automata~\cite{gkat} is defined to be the smallest class containing
\begin{enumerate}
	\item[(a)] every finite discrete coalgebra, and
	\item[(b)] \((\X + \Y)[X, h]\) whenever \(\X\) and \(\Y\) are well-nested.
\end{enumerate}

A short, relatively abstract proof of the following proposition was already given in \cref{sec:well_nested_coalgebras}.
We include the following more combinatorial proof as a supplement.

\restatewisnecessary*

\begin{proof}
	By induction on the construction of \(\V\).
	Of course, \(\V\) is discrete if and only if \(\V \models \coeq D\), so the base case follows from the definition of nestedness.

	For the inductive step, let \(\V = (\X + \Y)[X, h]\), where \(\X\) and \(\Y\) are well-nested coalgebras satisfying \(\coeq W\), and \(h : A \to 2 + \Sigma\times(X+Y)\).
	By finality, we obtain three homomorphisms
	\begin{align*}
		!_\X : \X \to \coalg Z, \ !_\Y : \Y \to \coalg Z, \ \text{ and } !_\V : \V \to \coalg Z.
	\end{align*}
	The first two satisfy \(!_\X[X], !_\Y[Y] \subseteq \coeq W\) by the induction hypothesis.
	Since \(Y\) is a subautomaton of \(\V\), \(!_\V(v) = !_\Y(v)\) for any \(v \in Y\), so it suffices to check that \(!_\V(v) \in \coeq W\) for \(v \in X\).
	To do this, we let \(!_\V(v) = t\) for an arbitrary \(v \in X\) and exhibit a construction of \(t\) from the nesting rules.

	We begin by showing the nestedness of \(t' :=\ !_{\V'}(v)\), where \(\V' := \X[X, h']\) and\[
		h'(a) := \begin{cases}
			1 &\text{ if \(h(a) \in Y\)},\\
			h(a) &\text{ otherwise}.
		\end{cases}
	\]
	This allows us to write \(\V = (\V' + \Y)[X, h]\) and \(t = t' \cdot s_1\), where\[
		s_1(a) = \begin{cases}
			1 &\text{ if \(h(a) = 1\) or \(h(a) \in \Sigma\times X\)},\\
			0 &\text{ if \(h(a) = 0\)},\\
			p &\text{ if \(h(a) = (p,y) \in \Sigma\times Y\)}
		\end{cases}
		\ \text{ and }\ \partial_a s_1 = !_\Y(\pi_2\circ h(a)).
	\]
	Indeed, \(t(w) = t'(w)\) for any \(w \in A^+\) such that \(\delta^\V(v,w) \in X\); as well as for any \(w \in A^+\) such that \(w = w' a\), \(\delta^\V(v,w') \in X\), and \(v \Rightarrow_{\V'} a\).
	Thus, it suffices to see that \(\coalg V' \models \coeq W\), and by extension that \(t' \in \coeq W\).

	Towards the construction of \(t'\), let \(t_0 = !_\X(v)\), and define
	\[
		s_0(a) = \begin{cases}
			1 &\text{ if \(h(a) = 1\) or \(h(a) \in \Sigma\times Y\)},\\
			0 &\text{ if \(h(a) = 0\)},\\
			p &\text{ if \(h(a) = (p,x) \in \Sigma\times X\)}
		\end{cases}
		\ \text{ and }\ \partial_a s_1 = !_\X(\pi_2 \circ h(a)).
	\]
	By the induction hypothesis, \(t_0, s_0 \in \coeq W\).
	We claim that \(t_0 \rhd s_0 = t'\).

	To verify the claim, first let \[
		C = \{x \in X \mid v \to_{\X}^+ x\ \text{ and }\ (\exists a \in A)(x \Rightarrow_{\X} a \ \text{ and }\ h(a) \in X)\},
	\]
	where \((-)^+\) denotes transitive closure.
	If \(C = \emptyset\), then \(t_0 = t'\).
	Since this puts \(t' \in \coeq W\), it suffices to consider the case where \(C \neq \emptyset\).
	
	Assuming \(C \neq \emptyset\), define \[
		m = \min\{|w| \mid w \in A^+ \ \text{ and }\ \delta^\X(v,w) \in C\}.
	\]
	Note that \(d(t_0, t') \le 2^{-m}\) by design.

	Next, set \[
		B = \{x \in X \mid (\exists a \in A)(x \Rightarrow_{\X} a\ \text{ and }\ h(a) \in X)\}.
	\]
	Of course, \(C \subseteq B\), so \(C \neq \emptyset\) means \(B \neq \emptyset\) also.
	If \(\neg(x \to_{\V'}^+ y)\) holds for all \(x, y \in B\), then \(t_0 \cdot s_0^{\cdot n} = t_0 \cdot s_0\) for all \(n > 0\).
	This also means that \(t' = t_0 \cdot s_0\), so it suffices to consider the case where \(x \to_{\V'}^+ y\) holds for some \(x, y \in B\).

	Assuming \(x \to_{\V'}^+ y\) holds for some \(x, y \in B\), let\[
		\rho = \min\{|w| \mid w \in A^+\ \text{ and }\ (\exists x,y \in B)(\delta^{\V'}(x,w) = y)\}.
	\]
	Every path of the form \[
		v \to_{\X}^+ x_0 \to_{\V'}^+ x_1 \to_{\V'}^+ \cdots \to_{\V'}^+ x_n
	\]
	with \(x_0 \in C\) and \(x_i \in B\) for \(i > 0\) is of length at most \(m + n\rho\).
	If each path is chosen to be the shortest possible path, then since a branch of \(t_0\) witnesses the path \(v \to_{\X}^+ x_0\), a branch of \(t_0 \cdot s_0\) witnesses the path \(x_0 \to_{\V'}^+ x_1\), and so on, we have \[
		d(t', t_0 \cdot s_0^{\cdot n}) \le 2^{-(m + n\rho)} \le 2^{-n\rho}.
	\]
	Hence, \(t_0 \rhd s_0 = \lim_{n \to \infty} t_0 \cdot s_0^{\cdot n} = t'\).
\end{proof}

% *** SECTION *** %
\section{Detailed proofs for \cref{sub:uniqueness_of_solutions_for_salomaa_systems}: Uniqueness of solutions for Salomaa systems}

Recall that any finite product of compact spaces is compact.
In particular, \(Z^n\) is compact for any \(n \in \N\).
Compact metric spaces are necessarily complete, so \(Z^n\) is complete as well.
This gives us access to the \emph{Banach fixed-point theorem}, which states that any function \(f : M \to M\) from a complete metric space \((M, d_M)\) to itself that satisfies \[
	(\exists z \in [0,1))(\forall x, y \in M)\ d_M(x,y) \le z d_M(f(x), f(y))
\]
has a unique fixed-point.
In the formula above, any \(z \in [0,1)\) witnessing this property is called a \define{contraction scalar} for \(f\).

\restatesolvinginZ*
\begin{proof}
	Let \(f : Z^n \to Z^n\) be the function defined component-wise by \[
		f(\mathbf t)_i = s_{i1}\cdot t_1 +_{b_{i1}} \cdots +_{b_{i(n-1)}} s_{in}\cdot t_n +_{b_n} c_i
	\]
	where \(\mathbf t = (t_i)_{\le n} \in Z^n\).
	We are going to show that \(f\) is a contraction mapping in the product metric
	\[
		d_p(\mathbf t, \mathbf t') :=  \max\{d(t_i, t_i') \mid i \le n\}
	\]
	on \(Z^n\), with contraction scalar \(1/2\), and deduce the result from the Banach fixed-point theorem.

	To this end, let \(\mathbf t, \mathbf t' \in Z^n\) be two \(n\)-tuples of trees, and fix an index \(i \le n\).
	% Here, we will show that \(f(\mathbf t)_i\) and \(f(\mathbf t')_i\) agree on all prefixes of words of the form \(awa'w'\), where \(a \in b_{ij}\), \(awa'\) is an accepting leaf of \(s_{ij}\) if it is in the domain of \(s_{ij}\), and \(t_j\) and \(t_j'\) agree on the prefixes of \(a'w'\).
	Clearly,\[
		d(f(\mathbf t)_i, f(\mathbf t')_i) = \max\{d(s_{ij} \cdot t_j, s_{ij} \cdot t_j') \mid j \le n\},
	\]
	since any word \(aw \in \dom(f(\mathbf t)_i)\cap \dom(f(\mathbf t')_i)\) at which \(f(\mathbf t)_i(aw) \neq f(\mathbf t')_i(aw)\) must begin with an atom \(a \in b_{ij}\) for some \(j \le n\).
	We argue below that, in fact, \(d(s_{ij} \cdot t_j, s_{ij} \cdot t_j') \le (1/2)d(t_j, t_j')\) for any \(j \le n\).
	It follows from this observation that 
	\[
		d(f(\mathbf t)_i, f(\mathbf t')_i) 
		= \max\{d(s_{ij} \cdot t_j, s_{ij} \cdot t_j') \mid j \le n\} 
		\le (1/2)\max\{d(t_j, t_j') \mid j \le n\} = (1/2)d(\mathbf t, \mathbf t'),
	\]
	which by definition of the product metric makes \(1/2\) a contraction scalar for \(f\).

	In general, \(d(t,t') \le d(s\cdot t, s \cdot t')\) for any \(s,t,t' \in Z\), and \(d(s,t) \le (1/2) d(\partial_a s, \partial_a t)\) when both derivatives are defined.
	Thus, for a fixed \(j \le n\) and atom \(a \in b_{ij}\), if \(s_{ij}(a) \in \Sigma\), we obtain \[
		d(s_{ij} \cdot t_j, s_{ij} \cdot t_j') 
		\le (1/2)d(\partial_a (s_{ij} \cdot t_j), \partial_a (s_{ij} \cdot t_j))
		= (1/2) d(\partial_a s_{ij} \cdot t_j, \partial_a s_{ij} \cdot t_j')
		\le (1/2) d(t_j, t_j').
	\]
	If there is no such atom, then \(s_{ij} = 0\), because \(s_{ij}\) is productive.
	This would then imply that \[
		d(s_{ij} \cdot t_j, s_{ij} \cdot t_j') = d(0,0) = 0 \le (1/2) d(t_j, t_j').
	\]
	In either case, \(d(s_{ij} \cdot t_j, s_{ij} \cdot t_j') \le (1/2)d(t_j, t_j')\) as desired.

	By definition of the product metric,
	\[
		d_p(f(\mathbf t), f(\mathbf t')) = \max\{d(f(\mathbf t)_i, f(\mathbf t')_i) \mid i \le n\}
		\le (1/2) d(\mathbf t, \mathbf t').
	\]
	Whence, \(f\) is a contraction map with contraction scalar \(1/2\).
	By the Banach fixed-point theorem, \(f\) has a unique fixed-point in \(Z^n\).
	This fixed-point is the unique \(\mathbf r \in Z^n\) satisfying
	\[
		r_i = s_{i1}\cdot r_1 +_{b_{i1}} \cdots +_{b_{i(n-1)}} s_{in}\cdot r_n +_{b_n} c_i,
	\]
	for all \(i \le n\).
\end{proof}

% *** SECTION *** %
\section{Detailed proofs for \cref{sec:a_completeness_theorem_for_gkat-}: Completeness w.r.t. \(\equiv_0\)}\label{appendix:detailed_proofs_for_cref_sec_a_completeness_theorem_for_gkat_}

% \begin{lemma}
% 	Let \(t \in Z\) and \(w \in \Node(t)\). Then
% 	\begin{enumerate}
% 		\item[(i)] \(\Node(\partial_w t) = \partial_w \Node(t)\), where \(\partial_w L := \{w' \mid ww' \in L\};\)
% 		\item[(ii)] \(\dom(t) = \Node(t)A\), where \(LA := \{wa \mid w \in L \wedge a \in A \};\)
% 		\item[(iii)] \(\dom(\partial_w w) = \partial_w \dom(t)\);
% 		\item[(iv)] \(\Leaf(t) = \dom(t) \setminus \Node(t)\); and
% 		\item[(v)] \(\Leaf(\partial_w t) = \partial_w \Leaf(t)\).
% 	\end{enumerate}
% \end{lemma}

% A congruence \(\gequiv\) on \(\Exp\) satisfying the axioms that generate \(\equiv_0\) is called a \define{\acro{GKAT}-equivalence}.

To prove the completeness theorem for \(\equiv_0\), we need the following lemma, which is a way of saying that \(e\) has finitely many derivatives.

\begin{lemma}\label{lem:Exp is locally-finite}
	The \GKAT-automaton \(\coalg E = (\Exp, D)\) is \define{locally finite}, meaning that for any \(e \in \Exp\), the subatomaton generated by \(e\), \(\langle e\rangle_{\coalg E}\), has finitely many states.
\end{lemma}

\begin{proof}
	Let \(|\langle e \rangle_{\coalg E}|\) be the cardinality of the set of states in the subatomaton \(\langle e \rangle_{\coalg E}\) of \(\coalg E\), and define \(\# : \Exp \to \N\) inductively as follows:
	\begin{mathpar}
		\#(b \subseteq A) = 1
		\and
		\#(p \in \Sigma) = 2
		\and
		\#(e +_b f) = \#(e) + \#(f)\\
		\#(e\cdot f) = \#(e) + \#(f) 
		\and
		\#(e^{(b)}) = \#(e)
	\end{mathpar}
	We will show that \(|\langle e \rangle_{\coalg E}| \le \#(e)\) for all \(e \in \Exp\), by induction on the construction of \(e\).

	Observe that if \(e = b\subseteq A\) or \(e = p \in \Sigma\), then \(|\langle e \rangle_{\coalg E}| = \#(e)\) by definition.
	This handles the base case.

	For the inductive step, assume \(|\langle e \rangle_{\coalg E}| \le \#(e)\) and \(|\langle f \rangle_{\coalg E}| \le \#(f)\), and let \(b \subseteq A\).
	Every syntactic derivative of \(e +_b f\) is a derivative of either \(e\) or \(f\), so immediately we obtain
	\[
		|\langle e +_b f \rangle_{\coalg E}| \le |\langle e \rangle_{\coalg E}| + |\langle f \rangle_{\coalg E}| \le \#(e) + \#(f) = \#(e +_b f).	
	\]
	Similarly, every derivative of \(e\cdot f\) is either of the form \(e' \cdot f\) for some derivative \(e'\) of \(e\), or is a derivative of \(f\).
	Hence, 
	\[
		|\langle e \cdot f \rangle_{\coalg E}| \le |\langle e \rangle_{\coalg E} \times \{f\}| + |\langle f \rangle_{\coalg E}| \le \#(e) + \#(f) = \#(e \cdot f).	
	\]
	Finally, every derivative of \(e^{(b)}\) is of the form \(e' \cdot e^{(b)}\) for some derivative \(e'\) of \(e\).
	These are in one-to-one correspondence with the derivatives of \(e\), so \(|\langle e^{(b)} \rangle_{\coalg E}| \le |\langle e \rangle_{\coalg E}| \le \#(e) = \#(e^{(b)})\).
\end{proof}

It follows from this lemma and \cref{prop:existence} that \(\coeq W\) is locally finite as well: indeed, if \(t = \sem e\), then \(\langle t \rangle_{\Z}\) is a subatomaton of the image of \(\langle e \rangle_{\coalg E}\) under \(!_{\coalg E}\) (in fact, the two are equal).
Thus, since \(\langle e \rangle_{\coalg E}\) is finite, so must \(\langle t \rangle_{\coalg E}\) be.

Now, we know that every finite automaton \(\X = (X, \delta)\) gives rise to a Salomaa system of left-affine equations\[
	S(\X) = \{x_i = e_{i1}\cdot x_1 +_{b_{i1}} \cdots +_{b_{i(n-1)}} e_{in}\cdot x_n +_{b_n} c_i \mid i \in I\},
\]
where \(X = \{x_i \mid i \in I\}\) is treated as a set of indeterminates, and
\[
\begin{aligned}
	e_{ij} = \bigplus_{x_i \trans{a|p_a} x_j} p_a,
\end{aligned}
 \quad\text{ where }\quad
 \begin{aligned}
	c_{i} &= \{a \in A \mid x_i \Rightarrow a\},\\
	b_{ij} &= \{a \in A \mid x_i \trans{a|p} x_j\}, \text{ and }\\
	X &= \{x_i \mid i \le n\}.
\end{aligned}
\]
By \cref{lem:Exp is locally-finite}, every expression \(e \in \Exp\) gives rise to a finite subautomaton \(\langle e \rangle_{\coalg E}\) of \(\coalg E\).
By the fundamental theorem, the inclusion map \(\langle e \rangle_{\coalg E} \hookrightarrow \coalg E\) is a solution to \(S(\langle e \rangle_{\coalg E})\).
By the uniqueness axiom, this inclusion map is the unique solution to \(S(\langle e \rangle_{\coalg E})\) up to \(\equiv_0\).
This shows that whenever two automata are isomorphic, \(\langle e\rangle_{\coalg E} \cong \langle f\rangle_{\coalg E}\), we have \(e \equiv_0 f\), since \(S(\langle e\rangle_{\coalg E})\) and \(S(\langle f\rangle_{\coalg E})\) are the same up to a renaming of variables.
The following much stronger statement can be shown, which we use to prove completeness.

\begin{lemma}\label{lem:basically completeness}
	Let \(e, f \in \Exp\), and assume the uniqueness axiom for \(\equiv_0\).
	If \(e\) and \(f\) are bisimilar, then \(e \equiv_0 f\).
\end{lemma}

\begin{proof}
	We argue in a similar manner to the isomorphism case.
	Let \(\X = \langle e \rangle_{\coalg E}\) and \(\Y = \langle f\rangle_{\coalg E}\), and \(R \subseteq X \times Y\) be a bisimulation relating \(e\) and \(f\).
	We equip \(R\) with a \acro{GKAT}-automaton structure \(\R = (R, \delta^\R)\) by setting\[
		\delta^\R((x,y))(a) = \begin{cases}
			n &\text{if \(\delta^\X(x)(a) = \delta^\Y(y)(a) = n \in 2\)},\\
			(x',y') &\text{if \(\delta^\X(x)(a) = x'\) and \(\delta^\Y(y)(a) = y'\)}.
		\end{cases}
	\]
	Since \(R\) is a bisimulation, this is well-defined, and furthermore the projection maps \(R \xrightarrow{\pi_1} X\) and \(R \xrightarrow{\pi_2} Y\) are \acro{GKAT}-automaton homomorphisms.
	Consider the Salomaa system of equations \(S(\R)\), as well as the maps \(\phi_e,\phi_f : R \to \Exp\) defined by \(\phi_e(x)=x\) and \(\phi_f(y) = y\).
	We argue that \(\phi_e\) and \(\phi_f\) are solutions to \(S(\R)\), and conclude from the uniqueness axiom that \(x \equiv_0 y\) for any \((x,y) \in R\).
	In particular, \(e \equiv_0 f\).

	To see that \(\phi_e\) is a solution to \(S(\coalg R)\), let \(|R| = k\) and consider an equation \[
		(x_i, y_i) = e_{i1} \cdot (x_1,y_1) +_{b_{i1}} \cdots +_{b_{i(k-1)}} e_{ik} \cdot (x_k,y_k) +_{b_{ik}} c_i
	\]
	in \(S(\R)\).
	The map \(\phi_e\) takes this to the equation
	\[
		x_i = e_{i1} \cdot x_1 +_{b_{i1}} \cdots +_{b_{i(k-1)}} e_{ik} \cdot x_k +_{b_{ik}} c_i.
	\]
	Now, where \([j] = \{l \mid x_{l} = x_j\} = \{[j]_1, \dots, [j]_m\}\), \(b_{i[j]} = b_{i[j]_m}\), and 
	\[
		g_{[j]} := e_{i[j]_1} +_{b_{i[j]_1}} e_{i[j]_2} +_{b_{i[j]_2}} \cdots +_{b_{i[j]_{m-1}}} e_{i[j]_m},
	\]
	we see that the right-hand side is \(\equiv_0\)-equivalent to
	\begin{align*}
		&e_{i1} \cdot x_{[1]} +_{b_{i1}} \cdots +_{b_{i(k-1)}} e_{ik} \cdot x_{[k]} +_{b_{ik}} c_i\\
		&\equiv_0 \left(e_{i1} \cdot x_{[1]} +_{b_{i1}} \cdots +_{b_{i[1]_{m-1}}} e_{i[1]_m} \cdot x_{[1]} \right)
		+_{b_{i[1]}} \cdots +_{b_{i[k]}} c_{i}\\
		&\equiv_0 g_{i[1]} \cdot x_{[1]} +_{b_{i[1]}} \cdots +_{b_{i[k]}} g_{i[k]} \cdot x_{i[k]} +_{b_{i[k]}} c_{i}.
	\end{align*}
	The final expression is precisely the \(x_{[i]}\)'th equation in \(S(\X)\), since \(x_{[i]} \trans{a|p} x_{[j]}\) if and only if \((x_{[i]}, y) \trans{a|p} (x_{[j]}, y')\) for some \(y,y' \in Y\) such that \((x_{[i]}, y),(x_{[j]}, y') \in R\).
	Since \(\X \hookrightarrow \coalg E\) is a solution to \(S(\X)\), 
	\[
		x_{[i]} \equiv_0 g_{i[1]} \cdot x_{[1]} +_{b_{i[1]}} \cdots +_{b_{i[k]}} g_{i[k]} \cdot x_{i[k]} +_{b_{i[k]}} c_{i}.
	\]
	Since \(i\) was arbitrary, \(\phi_e\) is a solution to \(S(\R)\).
	Similarly, the same holds for \(\phi_f\).
	Thus, by the uniqueness axiom, \(e \equiv_0 f\).
\end{proof}

\restatecompletenessforgkatminus*
\begin{proof}
	From \cref{lem:final bisimilarity} and \cref{prop:eta is a bialgebra}, we see that \(\sem e = \sem f\) if and only if \(e\) and \(f\) are bisimilar.
	Thus, by \cref{lem:basically completeness}, \(e \equiv_0 f\).
\end{proof}

% *** SECTION *** % 
\section{Detailed proofs for \cref{sec:a_completeness_theorem_for_gkat}: Completeness w.r.t. \(\equiv\)}

The normalized semantics can be connected to \(\equiv\) with relative ease, allowing us to recover the \emph{partial completeness} result from~\cite{gkat}, albeit with a different proof.

\begin{lemma}%
\label{lem:partial-completeness}
Let \(e \in \Exp\).
If \(\sem{e}\) is dead, then \(e \equiv 0\).
\end{lemma}
\begin{proof}
A straightforward check verifies that
\[
	R = \{ (t \cdot 0, t) \mid \text{\(t \in Z\) is dead} \}
\]
is a bisimulation.
From this, we know that \(\sem{e} \cdot 0 = \sem{e}\), and therefore that \(\sem{e \cdot 0} = \sem{e}\).
By completeness of \(\equiv_0\) w.r.t. \(\sem{-}\), we then know that \(e \cdot 0 \equiv_0 e\).
Since \(e \cdot 0 \equiv 0\) and \(\equiv_0\) is contained in \(\equiv\), we can conclude that \(e \equiv 0\).
\end{proof}

Interestingly, the result above does not depend on the uniqueness axiom.
The following technical lemma describes the interaction between normalization and the other operators in trees.

\begin{lemma}%
\label{lem:normalization-vs-operators}
If \(s, t, r \in Z\) and \(b \in \operatorname{BExp}\), then

\vspace{-1em}\begin{mathpar}
(s +_b t)^\wedge = (s^\wedge +_b t^\wedge)^\wedge
\and
(s \cdot t)^\wedge = (s^\wedge \cdot t^\wedge)^\wedge
\and
(t \cdot 0)^\wedge = 0 = 0^\wedge
\and
(t^{(b)})^\wedge = ((t^\wedge)^{(b)})^\wedge
\end{mathpar}
Furthermore, if \(t^\wedge = (r \cdot t +_b s)^\wedge\) and \(r\) is such that \(r(a) \neq 1\) for all \(a \in A\), then \(t^\wedge = (r^{(b)} \cdot s)^\wedge\).
\end{lemma}
\begin{proof}
In all cases, a straightforward coinductive argument suffices.
\end{proof}

\restatehatWsatisfiesGKAT*
\begin{proof}
	We proceed by induction on \(\equiv\).
	In all base cases except S3, we know that \(e \equiv_0 f\); by \cref{thm:Z satisfies GKAT^-}, we then know that \(\sem{e} = \sem{f}\), and hence \(\sem{e}^\wedge = \sem{f}^\wedge\).
	For S3, we have \(\sem{e \cdot 0}^\wedge = \sem{0}^\wedge\) by the third equality in \cref{lem:normalization-vs-operators}.

	The inductive cases for reflexivity, symmetry and transitivity are straightforward.
	The case for congruence w.r.t.\ the operators follows by the equalities in \cref{lem:normalization-vs-operators}.
	
	Finally, in the inductive step for W3, let \(e,f,g \in \Exp\) and \(b \in \operatorname{BExp}\) with \(E(f) \equiv 0\) and \(e \equiv f \cdot e +_b g\).
	By induction, \(\sem{E(g)}^\wedge = 0\) and \(\sem{e}^\wedge = \sem{g \cdot e +_b h}^\wedge\).
	First, note that \(\sem{E(g)}^\wedge = \sem{E(g)}\).
	By an argument similar to the one in \cref{thm:Z satisfies GKAT^-}, we can conclude that \(\sem{g}(a) \neq 1\) for all \(a \in A\).
	Applying the final implication in \cref{lem:normalization-vs-operators}, we can conclude that \(\sem{e}^\wedge = \sem{g^{(b)} \cdot h}^\wedge = \sem{f}^\wedge\).
\end{proof}

To prove that \(\coeq{W}\) is closed under normalization (this is \cref{lem:closure under normalisation}), we prove something more general.
When \(\coeq{P} \subseteq Z\) and \(t \in Z\), we write \(t \filter \coeq{P}\) for the \emph{pruning} of \(t\) by \(\coeq{P}\), which removes all subtrees of \(t\) that are in \(\coeq{P}\).
This operator is defined coinductively.

\vspace{-1em}\begin{mathpar}
    (t \filter \coeq{P})(a) =
        \begin{cases}
        0 & t(a) \in \Sigma \wedge \partial_a t \in \coeq{P} \\
        t(a) & \text{otherwise}
        \end{cases}
        \and
    \partial_a (t \filter \coeq{P}) = (\partial_a t) \filter \coeq{P}
\end{mathpar}
Clearly, if \(\coeq{P}\) is the coequation of dead trees, then \(t^\wedge = t \filter \coeq{P}\).
We now claim that if \(t \in \coeq{W}\) and \(\coeq{P} \subseteq Z\), then \(t \filter \coeq{P} \in \coeq{W}\).

\begin{lemma}%
\label{lem:concatenation-vs-filter}
Let \(t, s \in Z\) and \(\coeq{P} \subseteq Z\) be a coequation.
Then
\[
    (s \cdot t) \filter \coeq{P} = (s \filter \coeq{P}_t) \cdot (t \filter \coeq{P})
    \qquad\text{where}\qquad
    \coeq{P}_t = \{ r \in Z \mid r \cdot t \in \coeq{P} \}
\]
\end{lemma}
\begin{proof}
We claim that
\[
    R = \{ ((s \cdot t) \filter \coeq{P}, (s \filter \coeq{P}_t) \cdot (t \filter \coeq{P})) \mid t, s \in Z, \coeq{P} \subseteq Z \} \cup \Delta_Z
\]
is a bisimulation.
As before, we need only check the pairs in the first part, since the diagonal is already a bisimulation.

For the initial conditions, let \(a \in A\).
There are several cases to consider.
\begin{itemize}
    \item
    If \((s \cdot t)(a) \in \Sigma\) and \(\partial_a (s \cdot t) \in \coeq{P}\), then \(((s \cdot t) \filter \coeq{P})(a) = 0\).
    We should prove that \(((s \filter \coeq{P}_t) \cdot (t \filter \coeq{P}))(a) = 0\).
    \begin{itemize}
        \item
        If \(s(a) \in \Sigma\), then \(\partial_a s \cdot t = \partial_a (s \cdot t) \in \coeq{P}\), and therefore \(\partial_a s \in \coeq{P}_t\).
        Thus, \((s \filter \coeq{P}_t)(a) = 0\).

        \item
        If \(s(a) = 1\) and \(t(a) \in \Sigma\), then \(\partial_a t = \partial_a (s \cdot t) \in \coeq{P}\).
        Thus, \((s \filter \coeq{P}_t)(a) = 1\) and \((t \filter \coeq{P})(a) = 0\).
    \end{itemize}
    In both of these cases, \(((s \filter \coeq{P}_t) \cdot (t \filter \coeq{P}))(a) = 0\).

    \item
    Otherwise, \(((s \cdot t) \filter \coeq{P})(a) = (s \cdot t)(a)\).
    We should prove that \(((s \filter \coeq{P}_t) \cdot (t \filter \coeq{P}))(a) = (s \cdot t)(a)\).
    \begin{itemize}
        \item
        If \(s(a) = 0\), then \((s \cdot t)(a) = 0 = (s \filter \coeq{P}_t)(a) = ((s \filter \coeq{P}_t) \cdot (t \filter \coeq{P}))(a)\).

        \item
        If \(s(a) = 1\), then \((s \filter \coeq{P}_t)(a) = 1\) and \((s \cdot t)(a) = t(a)\).
        It remains to prove that \((t \filter \coeq{P})(a) = t(a)\).
		On the one hand, if \(t(a) \in 2\), then \((t \filter \coeq{P})(a) = t(a)\) immediately.
		On the other hand, if \(t(a) \in \Sigma\), then \(\partial_a t = \partial_a (s \cdot t) \not\in \coeq{P}\).
		Thus, \((t \filter \coeq{P})(a) = t(a)\).

        \item
        If \(s(a) \in \Sigma\), then \(\partial_a s \cdot t = \partial_a (s \cdot t) \not\in \coeq{P}\), thus \(\partial_a s \not\in \coeq{P}_t\).
        We then derive
        \[
            ((s \filter \coeq{P}_t) \cdot (t \filter \coeq{P}))(a)
                = (s \filter \coeq{P}_t)(a)
                = s(a)
                = (s \cdot t)(a)
        \]
    \end{itemize}
\end{itemize}
For the coinductive step, let \(a \in A\) is such that \(((s \cdot t) \filter \coeq{P})(a) = ((s \filter \coeq{P}_t) \cdot (t \filter \coeq{P}))(a) \in \Sigma\).
There are two cases.
\begin{itemize}
    \item
    First, if \(s(a) = 1\), then we derive
    \begin{align*}
        \partial_a ((s \cdot t) \filter \coeq{P})
            &= (\partial_a (s \cdot t)) \filter \coeq{P} \\
            &= \partial_a t \filter \coeq{P} \\
            &\mathrel{R} \partial_a t \filter \coeq{P} \\
            &= \partial_a ((s \filter \coeq{P}_t) \cdot (t \filter \coeq{P}))
    \end{align*}

    \item
    Otherwise, if \(s(a) \in \Sigma\), then
    \begin{align*}
        \partial_a ((s \cdot t) \filter \coeq{P})
            &= (\partial_a (s \cdot t)) \filter \coeq{P} \\
            &= (\partial_a s \cdot t) \filter \coeq{P} \\
            &\mathrel{R} (\partial_a s \filter \coeq{P}_t) \cdot (t \filter \coeq{P}) \\
            &= (\partial_a (s \filter \coeq{P}_t)) \cdot (t \filter \coeq{P}) \\
            &= \partial_a ((s \filter \coeq{P}_t) \cdot (t \filter \coeq{P}))
            \qedhere
    \end{align*}
\end{itemize}
\end{proof}

\begin{lemma}%
\label{lem:continuation-vs-filter}
Let \(t, s \in Z\) and \(\coeq{P} \subseteq Z\).
Then
\[
    (s \rhd t) \filter \coeq{P} = (s \filter \coeq{P}^t) \rhd (t \filter \coeq{P}^t)
    \qquad\text{where}\qquad
    \coeq{P}^t = \{ r \in Z \mid r \rhd t \in \coeq{P} \}
\]
\end{lemma}
\begin{proof}
For the initial conditions, there are several cases.
\begin{itemize}
    \item
    If \((s \rhd t)(a) \in \Sigma\) and \(\partial_a (s \rhd t) \in \coeq{P}\), then \(((s \rhd t) \filter \coeq{P})(a) = 0\).
    We should prove \(((s \filter \coeq{P}^t) \rhd (t \filter \coeq{P}^t))(a) = 0\).
    \begin{itemize}
        \item
        If \(s(a) \in \Sigma\), then \(\partial_a s \rhd t = \partial_a (s \rhd t) \in \coeq{P}\), and therefore \(\partial_a s \in \coeq{P}^t\).
        Thus, \((s \filter \coeq{P}^t)(a) = 0\).

        \item
        If \(s(a) = 1\) and \(t(a) \in \Sigma\), then \(\partial_a t \rhd t \in \coeq{P}\), whence \(\partial_a t \in \coeq{P}^t\).
        Thus, \((s \filter \coeq{P}^t)(a) = 1\) and \((t \filter \coeq{P}^t)(a) = 0\).
    \end{itemize}
    In both of these cases, it follows that \(((s \filter \coeq{P}^t) \rhd (t \filter \coeq{P}^t))(a) = 0\).

    \item
    Otherwise, \(((s \rhd t) \filter \coeq{P})(a) = (s \rhd t)(a)\).
    We should prove that \(((s \filter \coeq{P}^t) \rhd (t \filter \coeq{P}^t))(a) = (s \rhd t)(a)\).
    \begin{itemize}
        \item
        If \(s(a) = 0\), then \((s \rhd t)(a) = 0 = (s \filter \coeq{P}^t)(a) = ((s \filter \coeq{P}^t) \cdot (t \filter \coeq{P}^t))(a)\).
        
        \item
        If \(s(a) = 1\), then \((s \rhd t)(a) = t(a)\) and \((s \filter \coeq{P}^t)(a) = 1\).
        It remains to prove that \((t \filter \coeq{P}^t)(a) = t(a)\).
        On the one hand, if \(t(a) \in 2\), then \((t \filter \coeq{P}^t)(a) = t(a)\) immediately.
		On the other hand, if \(t(a) \in \Sigma\), then \(\partial_a t \rhd t = \partial_a (s \rhd t) \not\in \coeq{P}\).
		In that case, \(\partial_a t \not\in \coeq{P}^t\) as well.
		But then \((t \filter \coeq{P}^t)(a) = t(a)\).

        \item
        If \(s(a) \in \Sigma\), then \(\partial_a s \rhd t = \partial_a (s \rhd t) \not\in \coeq{P}\).
        In that case, \(\partial_a s \not\in \coeq{P}^t\) as well.
        We then derive
        \[
            ((s \filter \coeq{P}^t) \rhd (t \filter \coeq{P}^t))(a)
                = (s \filter \coeq{P}^t)(a)
                = s(a)
                = (s \rhd t)(a)
        \]
    \end{itemize}
\end{itemize}
For the coinductive step, let \(a \in A\) such that \(((s \rhd t) \filter \coeq{P})(a) = ((s \filter \coeq{P}^t) \cdot (t \filter \coeq{P}^t))(a) \in \Sigma\).
There are two cases.
\begin{itemize}
    \item
    First, if \(s(a) = 1\), then we derive
    \begin{align*}
        \partial_a ((s \rhd t) \filter \coeq{P})
            &= (\partial_a (s \rhd t)) \filter \coeq{P} \\
            &= ((\partial_a t \rhd t) \filter \coeq{P}) \\
            &\mathrel{R} ((\partial_a t \filter \coeq{P}^t) \rhd (t \filter \coeq{P}^t)) \\
            &= ((\partial_a (t \filter \coeq{P}^t)) \rhd (t \filter \coeq{P}^t)) \\
            &= \partial_a ((s \filter \coeq{P}^t) \rhd (t \filter \coeq{P}^t))
    \end{align*}

    \item
    Otherwise, if \(s(a) \in \Sigma\), then
    \begin{align*}
        \partial_a ((s \rhd t) \filter \coeq{P})
            &= (\partial_a (s \rhd t)) \filter \coeq{P} \\
            &= ((\partial_a s \rhd t) \filter \coeq{P}) \\
            &\mathrel{R} ((\partial_a s \filter \coeq{P}^t) \rhd (t \filter \coeq{P}^t)) \\
            &= ((\partial_a (s \filter \coeq{P}^t)) \rhd (t \filter \coeq{P}^t)) \\
            &= \partial_a ((s \filter \coeq{P}^t) \rhd (t \filter \coeq{P}^t))
            \qedhere
    \end{align*}
\end{itemize}
\end{proof}

\begin{proposition}\label{prop:pruning is chill}
Let \(t \in \coeq{W}\).
Then for all \(\coeq{P} \subseteq Z\) it holds that \(t \filter \coeq{P} \in \coeq{W}\).
\end{proposition}
\begin{proof}
We proceed by induction on \(\coeq{W}\).
In the base, \(t \in \coeq{D}\), meaning \(t \filter \coeq{P} = t\).
For the inductive step, there are three cases.
\begin{itemize}
    \item
    If \(t \in \coeq{W}\) because \(\partial_a t \in \coeq{W}\) for all \(a \in A\) with \(t(a) \in \Sigma\), then by induction \(\partial_a (t \filter \coeq{P}) = \partial_a t \filter \coeq{P} \in \coeq{W}\) for all \(a \in A\) with \((t \filter \coeq{P})(a) \in \Sigma\).
    It then follows that \(t \filter \coeq{P} \in \coeq{W}\).

    \item
    If \(t \in \coeq{W}\) because \(t = s \cdot r\) for \(s, r \in \coeq{W}\), then by induction \(s \filter \coeq{P}_r, r \filter \coeq{P} \in \coeq{W}\).
    By definition of \(\coeq{W}\) and \cref{lem:concatenation-vs-filter}, we then have that \((s \cdot r) \filter \coeq{P} = (s \filter \coeq{P}_r) \cdot (r \filter \coeq{P}) \in \coeq{W}\).
    
    \item
    If \(t \in \coeq{W}\) because \(t = s \rhd r\) for \(s, r \in \coeq{W}\), then by induction \(s \filter \coeq{P}^r, r \filter \coeq{P}^r \in \coeq{W}\).
    By definition of \(\coeq{W}\) and \cref{lem:continuation-vs-filter}, we then have that \((s \rhd r) \filter \coeq{P} = (s \filter \coeq{P}^r) \rhd (r \filter \coeq{P}^r) \in \coeq{W}\).
    \qedhere
\end{itemize}
\end{proof}

\restateclosureundernormalization*

\begin{proof}
	Take \(\coeq{P}\) to be the set of dead trees in \cref{prop:pruning is chill}.
\end{proof}

%\restategkatdoesntseehats*

\begin{lemma}\label{lem:gkat doesnt see hats}
	Let \(e \in \Exp\), and \(e^\wedge\) be a normalized expression for \(e\).
	Assume the uniqueness axiom for \(\equiv\) and \(\equiv_0\).
	Then \(e^\wedge \equiv e\).
\end{lemma}

\begin{proof}
	Let \(e \in \Exp\), and \(\X = \langle e\rangle_{\coalg E}\) be the Brzozowski automaton for \(e\), where very derivative \(e'\) of \(e\) (including \(e\) itself) is a state \(x_{e'}\).
	%Since \(\coeq W\) is closed under normalisation, there is an \(e^\wedge \in \Exp\) such that \(\sem{e^\wedge} = \sem{e}^\wedge\).
	%Let \(h_1,\dots,h_l\) be the states of \(\X\) such that \(h_i \neq e\) and \(\sem{h_i}^\wedge = 0\).
	Define \(\X^\wedge = (X, \delta^\wedge)\) to be the \GKAT-automaton obtained from \(\X = (X, \delta)\) by setting
	\[
		\delta^\wedge(x_{e'}, a) = \begin{cases}
			0 & x_{e'} \trans{a|p}_{\X} x_{e''} \wedge \text{\(\sem{e''}\) is dead} \\
			\delta(x,a) &\text{otherwise}.
		\end{cases}
	\]
	This \GKAT-automaton is finite, and hence induces a (finite) Salomaa system \(S(\X^\wedge)\) where each variable \(x_{e'}\) has a linear constraint that can be written (up to \(\equiv_0\)-equivalence) as
	\[
		x_{e'} = 1 +_{E(e')} \bigplus_{x_{e'} \trans{a|p}_{\X^\wedge} x_{e''}} p \cdot x_{e''}
	\]
	We claim that if for \(x_{e'}\) we fill in the expression \(e'\), then this constitutes a solution in \(\Exp/{\equiv}\).
	After all, we can derive using the fundamental theorem, \cref{lem:partial-completeness} and S3 that
	\[
		e'
			\equiv 1 +_{E(e')} \bigplus_{e' \trans{a|p}_{\coalg{E}} e''} p \cdot e''
			\equiv 1 +_{E(e')} \bigplus_{\substack{e' \trans{a|p}_{\coalg{E}} e'' \\ {\text{\(\sem{e''}\) is not dead}}}} p \cdot e''
			\equiv 1 +_{E(e')} \bigplus_{x_{e'} \trans{a|p}_{\X} x_{e''}} p \cdot e''
	\]
	
	The rest of the proof works by arguing that if for each \(x_{e'} \in X\) we fill in \(e'^\wedge\), then we have another solution to the Salomaa system of \(\X^\wedge\) in \(\Exp/{\equiv}\).
	Thus, we obtain the desired equivalence \(e \equiv e^\wedge\) from the uniqueness axiom for \(\equiv\).

	To this end, we first show that if we fill in \(\sem{e'}^\wedge = \sem{e'^\wedge}\) for \(x_{e'} \in X\), we have a solution to \(S(\X^\wedge)\) in \(Z\).
	%Since \(\sem{x^\wedge} = \sem{x}^\wedge\), this would establish that in fact, \(x \mapsto \sem{x}^\wedge\) solves \(S(\X^\wedge)\) in \(\coeq W\).
	By the completeness theorem for \(\equiv_0\), filling in \(e'^\wedge\) for \(x_{e'}\) gives a solution to \(S(\X^\wedge)\) in \(\Exp/{\equiv_0}\).
	It can be shown by induction on the construction of \(\equiv_0\) that \(\equiv_0\ \subseteq\ \equiv\).
	Whence, this particular choice of variables constitutes a solution to \(S(\X^\wedge)\) in \(\Exp/{\equiv}\) as desired.
	To see that choosing \(\sem{e'}^\wedge\) constitutes a solution to \(S(\X^\wedge)\) in \(Z\), let \(x_{e'} \mapsto t_{e'}\) be the unique solution to \(S(\X^\wedge)\) in \(Z\).
	We show that
	\[
		R = \left\{ \left(\sem {e'}^\wedge, t_{e'}\right) \mid x \in X \right\}
	\]
	is a bisimulation.
	%Let \(X^\wedge = \{x_1, \dots, x_m\}\), and consider the \(i\)th equation \[
	%	x_i = e_{i1} \cdot x_1 +_{b_{i1}} \cdots +_{b_{i(m-1)}} e_{im} \cdot x_m +_{b_{im}} c_i
	%\]
	%of the Salomaa system for \(\X^\wedge\).
	%Here, \(c_i = \{a \in A \mid x_i \Rightarrow_{\coalg E} a\}\), and
	Since \(t_{e'}\) is part of a solution to \(S(\X^\wedge)\) in \(Z\), we have \(\sem{e'}^\wedge(a) = 1 \iff \sem{e'}(a) = 1 \iff a \in E(e) \iff t_{e'}(a) = 1\).
	%, \(x_^\wedge \Rightarrow a\) for each \(a \in c_i\) as well.
	%Since \(t_{x_i}(a) = 1\) if and only if \(a \in c_i\), we see that \(t_{x_i}(a) = 1\) if and only if \(\sem{x_i}^\wedge(a) = 1\).
	On the other hand, 
	\begin{align*}
		t_{e'}(a) = 0
		&\iff \text{\(x_{e'} \trans{a|p}_{\X^\wedge} x_{e''}\) does not hold for any \(a\)} \\
		&\iff \text{\(x_{e'} \trans{a|p} x_{e''}\) and \(\sem{e''}\) is dead, or \(e' \downarrow a\)} \\
		&\iff \text{\(e \trans{a|p} e''\) and \(\sem{e''}\) is dead, or \(e' \downarrow a\)} \\
		&\iff \sem{e'}^\wedge(a) = 0.
	\end{align*}
	We are left with the coinductive step.
	In one direction, note that if \(t_{e'} \trans{a|p} \partial_a t_{e'}\), then \(\partial_a t_{e'} = t_{e''}\) with \(x_{e'} \trans{a|p} x_{e''}\), because the \(t_{e'}\) are a solution to \(S(\X^\wedge)\).
	In other words, \(\sem{e''}\) cannot be dead, and \(e' \trans{a|p} e''\).
	We find \[
		\partial_a \sem{e'}^\wedge
			= {(\partial_a \sem{e'})}^\wedge
			= {\sem{\partial_a e'}}^\wedge
			= \sem{e''}^\wedge.
	\]
	Conversely, if \(\sem{e'}^\wedge \trans{a|p} \partial_a \sem{e'}^\wedge\), then \(a\) is a node of \(\sem{e'}^\wedge\), which means that \(a\) is also a node of \(t_{e'}\) by the arguments above.
	Thus,
	\(
		\partial_a t_{e'}
			= t_{e''}
	\)
	where \(e' \trans{a|p} e''\), since the \(t_{e'}\) are a solution to \(S(\X^\wedge)\) in \(Z\).
	%This puts \(a \in b_{ij}\), and tells us that \(t_{x_i} \trans{a|p} t_{x_j}\), since \(\{t_x \mid x \in X^\wedge\}\) solves \(S(\X^\wedge)\).
	In either case, \((\partial_a \sem{e'}^\wedge, \partial_a t_{e'}) \in R\), so \(R\) is a bisimulation.
	By simplicity of \(Z\), \(\sem{e'}^\wedge = t_{e'}\) for all \(x_{e'} \in X\), and therefore \(x_{e'} \mapsto \sem{e'}^\wedge\) solves \(S(\X^\wedge)\) in \(Z\).
\end{proof}

\restatecompletenessforgkat*
\begin{proof}
	Since \(\sem{e}^\wedge = \sem{f}^\wedge\), also \(\sem{e^\wedge} = \sem{f^\wedge}\).
	By \cref{cor:completeness for GKAT^-,lem:gkat doesnt see hats}, we can then derive
	\[
		e \equiv e^\wedge \equiv_0 f^\wedge \equiv f
		\qedhere
	\]
\end{proof}

\end{document}